\documentclass[11pt]{article}
\usepackage[utf8]{inputenc}

\usepackage{amsmath, amssymb, amsthm}
\usepackage{fullpage}
\usepackage{graphics}
\usepackage{environ}
\usepackage{url}
\usepackage[hidelinks]{hyperref}
\usepackage{algorithm}
\usepackage[noend]{algpseudocode}
\usepackage[labelfont=bf]{caption}
\usepackage[bottom]{footmisc}
\usepackage{natbib}

\usepackage{framed}
\usepackage[framemethod=tikz]{mdframed}
\usepackage{appendix}
\usepackage{graphicx}
\usepackage{mathtools}
\usepackage{enumitem}
\usepackage{graphicx}
\usepackage{subcaption}

\newtheorem{theorem}{Theorem}

\newtheorem{lemma}[theorem]{Lemma}
\newtheorem{corollary}[theorem]{Corollary}

\theoremstyle{remark}
\newtheorem{remark}{Remark}

\newcommand{\D}{\mathcal{D}}
\newcommand{\V}{\mathcal{V}}
\newcommand{\tV}{\tilde{\V}^{(2)}}
\newcommand{\tVo}{\tilde{\V}}
\newcommand{\tPi}{\tilde{\Pi}^{(2)}}
\newcommand{\tPio}{\tilde{\Pi}}
\newcommand{\vb}{\mathbf{v}}
\newcommand{\eps}{\varepsilon}

\DeclareMathOperator{\E}{\mathbb{E}}
\def \Rev  {{\sf Rev}}

\def \supp {{\sf supp}}

\def \spike  {{\sf spike}}
\def \rnd  {{\sf rnd}}
\def \pad  {{\sf pad}}

\setcitestyle{acmnumeric}

\title{Prior-free Dynamic Mechanism Design With Limited Liability}

\begin{document}

\author{Mark Braverman\thanks{Department of Computer Science, Princeton University, email: mbraverm@cs.princeton.edu. Research supported in part by the NSF Alan T. Waterman Award, Grant
No. 1933331, a Packard Fellowship in Science and Engineering, and the
Simons Collaboration on Algorithms and Geometry. Any opinions,
findings, and conclusions or recommendations expressed in this
publication are those of the author and do not necessarily reflect the
views of the National Science Foundation.} \and Jon Schneider\thanks{Google Research, email: jschnei@google.com.} \and S. Matthew Weinberg\thanks{Department of Computer Science, Princeton University, email: smweinberg@princeton.edu. Supported by NSF CAREER Award CCF-1942497.}}
\date{}

\begin{titlepage}

\maketitle

\begin{abstract}
We study the problem of repeatedly auctioning off an item to one of $k$ bidders where: a) bidders have a per-round individual rationality constraint, b) bidders may leave the mechanism at any point, and c) the bidders' valuations are adversarially chosen (the prior-free setting). Without these constraints, the auctioneer can run a second-price auction to ``sell the business'' and receive the second highest total value for the entire stream of items. We show that under these constraints, the auctioneer can attain a constant fraction of the ``sell the business'' benchmark, but no more than $2/e$ of this benchmark. 

In the course of doing so, we design mechanisms for a single bidder problem of independent interest: how should you repeatedly sell an item to a (per-round IR) buyer with adversarial valuations if you know their total value over all rounds is $V$ but not how their value changes over time? We demonstrate a mechanism that achieves revenue $V/e$ and show that this is tight.
\end{abstract}

\end{titlepage}

\section{Introduction}\label{sec:intro}

Classical dynamic mechanism design traditionally assumes that the values of the participants are generated \textit{stochastically}: for example, drawn iid from some distribution every round, or generated according to a simple stochastic process known to the mechanism designer. These assumptions are often somewhat unrealistic; actual valuations drift over time and are subject to shocks in a way which is hard to predict or model. A mechanism designer may hope for a dynamic mechanism robust to these features of the problem. 

The goal of this paper is to initiate the study of dynamic mechanism design in the presence of \textit{adversarially} chosen valuations (the prior-free setting). Specifically: there are $k$ bidders, and $T$ rounds. Bidder $i$ has value $v_{i,t}$ for the round-$t$ item, and bidders are additive across rounds. The designer's goal is to maximize their \emph{revenue}. 

It is of course impossible to achieve any approximation guarantee with respect to the ``first-best'' benchmark of $\sum_t\max_i \{ v_{i,t}\}$ (perhaps all but one bidder has a constant value of $0$ --- there is no way to guess the correct price to set for the remaining bidder). This challenge also arises in the vast literature on prior-free auctions in other domains, and care is required to pick an insightful ``second-best'' benchmark~\cite{GoldbergHW01,GoldbergHKS04,GoldbergHKSW06,ChenGL14,DevanurHY15,ChenGL15}.

Let us quickly consider two potential benchmarks. First, one might target the second-best per-round: $\Rev_{SPA}:=\sum_t \text{Second-Highest}_i \{v_{i,t}\}$. There is also a simple, dominant-strategy truthful auction which achieves $\Rev_{SPA}$: simply auction each item via a second-price auction. This benchmark, however, may be infinitely far from optimal if for each item there is a unique bidder with non-zero value. On the other extreme, one might target the aggregate second-best: $\Rev_{STB}:=\text{Second-Highest}_i \{\sum_t v_{i,t}\}$. There is also a simple, dominant-strategy truthful auction which achieves $\Rev_{STB}$: simply ``sell the business'' and auction the entire stream of items at the very beginning.\\

\noindent\textbf{Limited Liability and Per-Round Individual Rationality.} In a static setting, $\Rev_{SPA}$ corresponds to selling the items separately (using a second-price auction), and $\Rev_{STB}$ corresponds to bundling the items together (for a second-price auction). Both mechanisms are simple to implement in the static setting, and in fact variants of (the better of) these are known to yield good approximation guarantees in static settings~\citep{BabaioffILW14, Yao15, ChawlaM16, CaiDW16, CaiZ16}. In a dynamic setting, however, bundling all time slots together requires the bidder to pay a large amount up front and trust that seller will honor their promise to follow the protocol in the rounds to come. For this reason, much literature on dynamic mechanism design studies what is possible under per-round ``limited liability'' or ``individual rationality'' constraints on the buyer, where they never pay more than their value per round \citep{PapadimitriouPPR16, AshlagiDH16, MirrokniLTZ16, MirrokniLTZ16b, BMP16, MLTZ18, ADMS18}. Recalling that the benchmark $\Rev_{STB}$ can be achieved in the static setting (or the dynamic setting without limited liability), the main question we ask is: 

\begin{quote}
\emph{What approximation to $\Rev_{STB}$ can be guaranteed by a dominant-strategy truthful mechanism in the dynamic setting with limited liability?}
\end{quote}

Observe further that any positive resolution to the above question will also imply an approximation guarantee to $\max\{\Rev_{STB},\Rev_{SPA}\}$, by simply flipping a fair coin to run either the designed mechanism or a per-round second-price auction (which already satisfies limited liability). In other words, the $\Rev_{STB}$ benchmark is key to understanding the revenue lost in prior-free settings \emph{specifically due to limited liability constraints} (rather than other sources which harm revenue already in the static setting). 

\subsection{Main results and techniques}

We design a dominant-strategy truthful, limited-liability mechanism which achieves a constant-factor approximation to $\Rev_{STB}$. We also establish that it is impossible to achieve an arbitrarily close approximation to $\Rev_{STB}$, let alone match it. Below are informal statements of our main results, along with pointers to the formal statements.

\begin{theorem}[Informal Restatement of Theorem \ref{thm:twopmech}]
There exists a mechanism (for two bidders) which guarantees revenue at least $\frac{1}{2e}\Rev_{STB} - O(1)$. 
\end{theorem}

It is also worth noting that our mechanism does \emph{not} require the bidders to know their individual per round values ahead of time, but just their average value over all $T$ rounds. Still, the mechanism is dominant-strategy truthful even if they have full knowledge of the future. For $k$ bidders, a simple generalization of our 2-bidder mechanism leads to a competitive ratio of $1/(ke)$. We also show how to improve this to a constant competitive ratio independent of $k$, although the solution concept for our auction is implementation in undominated strategies (that is, it guarantees an approximation guarantee of $\alpha$ \emph{as long as every bidder plays an undominated strategy}).

\begin{theorem}[Informal Restatement of Theorem \ref{thm:kbiddermech}]
There exists a mechanism for $k$ bidders which guarantees revenue at least $\alpha \cdot \Rev_{STB} - O(k)$ for some constant $\alpha$ (independent of $k$). 
\end{theorem}

Finally, we show it is impossible to achieve revenue arbitrarily close to $\Rev_{STB}$, even when there are only two bidders. 

\begin{theorem}[Informal Restatement of Theorem \ref{thm:2plb}]
For any $\alpha > 2/e$, there is no mechanism (for two bidders) which always guarantees revenue at least $\alpha\cdot\Rev_{STB} - O(1)$. 
\end{theorem}

Since the $k$ bidder game contains the $2$ bidder game as a subcase, it is still the case that no mechanism for $k$ bidders can achieve an approximation ratio better than $2/e$ (Corollary \ref{cor:kbidderub}).\\

\noindent\textbf{Technical Overview: A Related Single-Bidder Problem.} The main ingredient in our above results is actually a thorough understanding of a related single-buyer problem. Specifically, consider exactly the same setup as our model, except there is only a single bidder \emph{and the designer knows $V=\sum_t v_t$} (or equivalent, that the designer knows $\nu:=V/T$). In the static setting, there is a trivial optimal solution (sell the bundle of all items for $V$, which is the first-best). Here we seek to answer the same questions as above: what fraction of this benchmark can the seller guarantee subject to limited liability?

Note that the prior-free nature of the valuations is an essential detail here. If the valuations were generated stochastically from a time-invariant distribution, the seller could run mechanisms of the form ``continue giving the item to the buyer as long as the buyer has paid at least $(1-\eps)\nu$ on average each round so far'' and extract almost the entire expected welfare. When valuations are adversarial, such mechanisms are no longer possible -- it could be the case that the value of the buyer is entirely back-loaded and is zero for the first half of the rounds. Nonetheless, we show that the seller can still (asymptotically) attain a $1/e$ fraction of the total welfare $V$, and that this is tight.  

\begin{theorem}[Informal Restatement of Theorems \ref{thm:mainalg} and \ref{thm:singlb}]
There is a mechanism for the single buyer game that achieves revenue $V/e - O(1)$. Moreover, this is tight: for any $\alpha > 1/e$, there is no mechanism for the single buyer game that can guarantee revenue $\alpha V - O(1)$. 
\end{theorem}

The mechanism we design for Theorem~\ref{thm:mainalg} is what we call a \textit{pay-to-play mechanism}. Such mechanisms are parameterized by an increasing function $\rho: [0, 1] \rightarrow (0, 1]$ and work as follows: on round $t$, the seller allocates the item to the buyer with probability $\rho(X_{t}/V)$, where $X_{t}$ is the total payment the buyer has made to the seller by time $t$. The buyer is free to make any payment they wish each round --- since $\rho$ is increasing, they are incentivized to pay early to increase their allocation probability later (and in particular, for whatever total payment they choose to make, they are incentivized to make that payment as soon as possible).

We show that the optimal such $f$ (specifically, $f(w) = \min(e^{ew-1}, 1)$) leads to a pay-to-play mechanism with revenue $V/e$. Our impossibility result shows that this mechanism is optimal, and uses much of the same machinery. One highlight of this approach is that we reparameterize \emph{any} truthful mechanism in terms of the allocation/payment \emph{as a function of the maximum possible welfare the buyer could have had so far} (rather than as a function of the round $t$). Under this reparameterization, we find a pay-to-play interpretation of any mechanism, meaning that the optimal pay-to-play mechanism is in fact optimal among all mechanisms.\\

\noindent\textbf{Technical Overview: Back to Multiple Bidders.} We apply the following technique to extend these ideas to 2 (and $k > 2$) bidders. We begin by soliciting each bidder's total value (e.g. $V_{a} = \sum a_{t}$ and $V_{b} = \sum b_{t}$). We then split each item in half and run two copies of the single bidder mechanism: for Alice, we run the single bidder mechanism on one of the halves, assuming her total value for these halves is $V_{b}/2$, and for Bob we run the single bidder mechanism on the other halves, assuming his total value for these halves is $V_{a}/2$. If $V_{a} > V_{b}$, then $V_{b}/2$ is an underestimate for Alice's value, and our single bidder mechanism guarantees we receive revenue at least $V_{b}/2e$, or equivalently, $\Rev_{STB}/2e$. 

This simple idea extends to $k$ bidders, but causes the approximation guarantee to degrade to $1/(ke)$. The main issue with this approach is that we are allocating $1/k$ of the items to each bidder, even to bidders that cannot afford them. To fix this, we introduce an element of competition to this mechanism in the form of a first-price auction. Roughly speaking, each round $t$ we allow each bidder $i$ to submit a bid $b_{i,t}$ in an attempt to increase the rate at which they are allocated items. The bidder with the highest bid pays that amount, which in turn increases the rate at which they are allocated items (this description elides several details which can be found in Section \ref{sect:kbidders}). This mechanism is not truthful, but we show that \emph{as long as all bidders play undominated strategies}, this mechanism achieves an $\alpha$-approximation to $\Rev_{STB}$ for a constant $\alpha$ independent of $k$. 

\subsection{Related Work}

The problem of selling a stream of items to one of several bidders is a central problem in the field of dynamic mechanism design (for a general introduction to the area, we recommend the survey \cite{bergemann2018dynamic}). Our treatment of this problem imposes two important constraints: the assumption of lack of long-term trust between the auctioneer and seller, and the adversarial prior-free nature of the bidders' valuations. 

\paragraph{Lack of long-term commitment.}
There is a major recent line of work which studies the design of dynamic mechanisms satisfying ex-post IR constraints (or ``limited liability'')~\citep{PapadimitriouPPR16, AshlagiDH16, MirrokniLTZ16, MirrokniLTZ16b, BMP16, MLTZ18, ADMS18}. In this language, our constraint that bidders never pay more than their value in a given round is a ``per round ex-post IR'' constraint. Most of these papers focus on a specific class of ex-post IR mechanisms called \textit{bank account mechanisms}, where the seller maintains a virtual ``bank account'' for each bidder, and allows bidders to draw from and pay to this bank account so that they end up paying (for example) their average value per round. Such mechanisms work well for stochastic valuations (where the total liability of the seller -- the size of the bank account -- can be bounded sublinearly in $T$) but fail for adversarial valuations (where the liability of the seller can grow linearly in $T$, e.g. for a bidder whose values are back-loaded). 

\paragraph{Prior-free valuations.}
There is a large spectrum of ways to model the valuations of the buyer, from weakest (being independently drawn each round from a time-invariant distribution) to strongest (the prior-free setting). The majority of the work mentioned above studies settings where each bidder's value is drawn iid each round from some distribution known to the seller (or occasionally more generally, the value for round $t$ is drawn from a distribution specific to round $t$ but also known to the seller in advance). One notable exception is \cite{MLTZ18} (and the follow-up work \cite{deng2019robust}), which studies ``non-clairvoyant dynamic mechanism design'' for selling to a single buyer, where the buyer has a different value distribution for each round but the seller is only aware of the distribution for the current round. 

One related line of work is the line of work on dynamic mechanisms for buyers with evolving values \citep{BB84, Besanko85, AtheyB01, KLN13, BS15, CDKS16}. In these works, the value of the buyer evolves according to some stochastic process over time (perhaps depending on the action of the seller or buyer). Again, the standard assumption here is that the seller has detailed knowledge of this stochastic process, and as such tend not to apply to adversarial valuations. 

To the best of our knowledge, very few papers have studied dynamic mechanism design in a completely adversarial setting. The most closely related such work is \cite{deng2019prior}, which studies the problem of selling to a buyer with adversarial valuations, but makes the additional assumption that the buyer is running a low-regret learning algorithm. In contrast to this, prior-free \emph{static} mechanism design has been extensively studied~\citep{GoldbergHW01,FiatGHK02, GoldbergHKS04,GoldbergHKSW06,ChenGL14,DevanurHY15,ChenGL15}. See also Chapter 7 of~\cite{hartline2013mechanism}.

\paragraph{Contract Theory.}
The primary benchmark we compare against is the revenue obtained when ``selling the business''; i.e. selling the entire stream of items at the beginning of the protocol through a second-price auction. A similar strategy of ``selling the firm'' is a common concept in the contract theory literature. As in our setting, for the classical principal-agent problem with hidden actions, ``selling the firm'' is the optimal revenue strategy for the principal in the absence of limited liability. \cite{dutting2018simple} studies the approximation ratio between the optimal contract and this first-best benchmark in this classical problem.

Indeed, our setting can be viewed as a peculiar multi-agent dynamic contract problem, where the principal is the seller and the buyers are agents (with actions corresponding to possible payments to the principal). A similar approach is taken in \cite{braverman2019multi}, which examines this model (in the stochastic values setting) when the principal uses a low-regret learning algorithm to allocate items. 

\section{Model}

\subsection{Multiple bidder game}\label{sect:modelmulti}

We consider a setting with $k \geq 2$ bidders and one seller. There are $T$ rounds, and in each round the seller has one item for sale. Bidder $i$ has value $v_{i, t} \in [0, 1]$ for the item sold at time $t$. Unlike in much prior work, where $v_{i, t}$ are generated stochastically from some known prior, we consider the case where all $v_{i, t}$ are adversarially set at the beginning of the game. We assume the seller has no knowledge of the valuations $v_{i, t}$. We further assume that each bidder $i$ has full knowledge of their own valuations $v_{i, t}$ over all times $t$. This assumption exists largely so that strategic play is well-defined for each bidder; in many of the mechanisms we present, we will see that it is possible to relax this assumption to bidders who only know their average valuation per round. 

During each round $t$, the seller and the $k$ bidders are allowed to participate in arbitrary communication. At the end of round $t$, the seller allocates the item among the $k$ bidders by giving some fraction\footnote{For mathematical convenience, we assume the seller runs a deterministic strategy with the ability to fractionally allocate the item. We believe our results should carry over (with perhaps some loss in terms sublinear in $T$) to randomized settings where the seller awards the entire item to bidder $i$ with probability $r_{i, t}$.} $r_{i, t}$ of the item to bidder $i$ (with $\sum_{i} r_{i, t} \leq 1$; in particular, the seller does not need to allocate all of the item). In return for this portion of the item, bidder $i$ pays the seller some price $p \in [0, r_{i, t}v_{i, t}]$. Note that we do not include any mechanism for the seller to force the bidder to commit to paying a specific price $p$ for the item in a given round; but as this is a repeated game, informal commitments of the form ``if bidder $i$ does not pay the agreed upon price in a round, the seller will never allocate any item to bidder $i$ ever again'' are possible\footnote{In particular, all our results extend to an alternate model where the seller offers bidder $i$ a price $p$ for some fraction of the item, and the buyer either can either accept the price and take the item or reject the price and leave the item.}.

Note also that we assume the bidder never pays the seller more than their value for the item in a given round. This can be seen as an ``ex post individual rationality'' or ``limited liability'' assumption, and rules out protocols where the seller simply sells the entire stream of items at the beginning of the protocol to the highest bidder. Nonetheless, we will see that our mechanisms have good performance even when some or all of the bidders are allowed to break this constraint and pay more than their value. 

Finally, at the end of the game, we allow a single bonus round where the seller is allowed to reimburse each bidder an amount of the seller's choosing. These reimbursements provide a convenient way to incentivize bidders to participate for all $T$ rounds (otherwise e.g. there is never any incentive to pay the seller anything in round $T$), but are otherwise unimportant and are small (on the order of $O(1)$) in most of the mechanisms we present\footnote{One notable exception is in Section \ref{sect:kbidders}, where we first demonstrate a mechanism for the $k$ bidder case which uses reimbursements on the order of $O(T)$. We subsequently show how to reduce the size of these reimbursements to $O(1)$, but some may find the mechanism with large reimbursements more natural.}. 

\paragraph{Strategies and solution concepts.} We are primarily interested in mechanisms  with dominant strategy equilibria for the bidders. Specifically, let $s_i$ denote a deterministic strategy for bidder $i$ (i.e., a function from states of the mechanism and the valuations $v_{i, t}$ to allowed actions of bidder $i$ at that state), and let $U_i(s_i, s_{-i})$ denote the expected utility received by bidder $i$ (where the expectation is over potential randomness in the seller's actions) when $i$ plays according to $s_i$ and all other bidders play according to $s_{-i}$. Then the tuple $(s^*_1, s^*_2, \dots, s^*_k)$ is a dominant strategy equilibrium for this mechanism if for all $s_i$ and $s_{-i}$,

$$U_i(s^*_i, s_{-i}) \geq U_i(s_i, s_{-i}).$$

Any mechanism with a dominant strategy equilibrium can be transformed into a truthful direct-revelation mechanism via the revelation principle, where the bidders report their private types (their valuations $v_{i, t}$) to the seller at the beginning of the protocol and the seller uses this to simulate the bidders' dominant strategies $s^*_i$ on behalf of the bidders. In our setting, we specify a direct-revelation mechanism as follows.

A \textit{direct-revelation mechanism} for the multiple bidder game specifies the following information for each type profile $\mathbf{v} = \{v_{i, t}\}_{i \in [k], t \in [T]}$ of all $k$ bidders:

\begin{itemize}
    \item For each bidder $i$, the fraction $r_{\mathbf{v}, i, t}$ of the item allocated to bidder $i$ at time $t$. (These must satisfy $\sum_{i} r_{\mathbf{v}, i, t} \leq 1$ for any $\mathbf{v}$ and $t$).
    \item For each bidder $i$, the price $x_{\mathbf{v}, i, t}$ the seller expects bidder $i$ to pay on round $t$.
    \item For each bidder $i$, the reimbursement $g_{\mathbf{v}, i}$ paid to bidder $i$ at the end of the game.
\end{itemize}

\noindent
To run a direct-revelation mechanism, the seller proceeds as follows:

\begin{enumerate}
    \item The seller begins by soliciting the type profile $\mathbf{v}$ from all of the bidders. The seller computes $r_{\mathbf{v}, i, t}$, $p_{\mathbf{v}, i, t}$, and $g_{\mathbf{v}, i}$, and distributes this information to the bidders. 
    \item On round $t$, the seller allocates to bidder $i$ a fraction  $r_{\mathbf{v}, i, t}$ of the item (note that there can be some fraction of the item that goes unallocated), unless bidder $i$ has been eliminated (in which case bidder $i$ receives nothing).
    \item This bidder responds by paying the seller some price $p$. If $p \neq x_{\mathbf{v}, i, t}$, then the seller eliminates bidder $i$ from further consideration (i.e. will never allocate to or reimburse bidder $i$ again).
    \item Finally, at the end of the game, the seller gives a reimbursement of $g_{\mathbf{v}, i}$ to each uneliminated bidder $i$.
\end{enumerate} 

In order for a direct-revelation mechanism to be \textit{truthful}, these allocations, prices, and reimbursements must satisfy the following two constraints:

\begin{itemize}
    \item \textit{(Limited liability)} For any $\mathbf{v}$, $i$, and $t$, we must have $x_{\mathbf{v}, i, t} \leq r_{\vb, i, t}v_{i, t}$ (i.e., a limited-liability bidder reporting truthfully must be able to pay for the item if they win in this round).
    \item \textit{(Incentive compatibility)} Fix a type profile $\mathbf{v}$ and a bidder $i$, and let $\mathbf{v}'$ be the profile where bidder $i$'s type $v_{i, t}$ is replaced by a new type $v'_{i, t}$. Then, if it is the case that
    
    $$x_{\mathbf{v}', i, t} \leq r_{\vb', i, t}v_{i, t}$$
    
    for all $t \leq \tau$ (i.e. our bidder with limited liability can pretend to have type $\mathbf{v}'$ for $\tau$ rounds) we must have that, for any $\tau \leq T$, 
    
    $$\left(\sum_{t=1}^{T} r_{\mathbf{v}, i, t}v_{i, t} - x_{\mathbf{v}, i, t}\right) + g_{\mathbf{v}, i} \geq \left(\sum_{t=1}^{\tau} r_{\mathbf{v}', i, t}v_{i, t} - x_{\mathbf{v}', i, t}\right) + x_{\vb', i, \tau}.$$
    
    That is, a bidder cannot increase their expected utility by misreporting their type and then following the resulting mechanism for $\tau$ steps before defecting. The additional term on the RHS represents the fact that the bidder does not need to pay the payment $x_{\vb', i, \tau}$ if they defect on round $\tau$. Note that both $x_{\vb', i, \tau}$ and $g_{\vb, i}$ will typically be $O(1)$, in comparison to the two sums which will typically be $\Theta(T)$.
\end{itemize}

\paragraph{Revenue and benchmarks.}

We wish to design truthful mechanisms with high \textit{revenue}: the total sum of payments paid to the seller, minus reimbursements. For a truthful direct-revelation mechanism and bidders with type profile $\mathbf{v}$, this can be written as

$$\Rev(\mathbf{v}) = \sum_{i = 1}^{k}\sum_{t=1}^{T} x_{\mathbf{v}, i, t} - \sum_{i=1}^{k} g_{\mathbf{v}, i}.$$

% For convenience, we will write $x_{\mathbf{v}, i, t} = r_{\mathbf{v}, i, t}p_{\mathbf{v}, i, t}$ to denote the expected payment of bidder $i$ to the seller in round $t$. We can then simply write $\Rev(\mathbf{v}) = \sum_{i, t} x_{\mathbf{v}, i, t} - \sum_{i} g_{\mathbf{v}, i}$.

We will compare the revenue the seller achieves to a benchmark function of $\mathbf{v}$. In this paper, we primarily consider the benchmark of the revenue obtained when auctioning off all $T$ items at once at the beginning of the game, i.e. ``selling the business''. This is not possible with limited-liability bidders, but with unconstrained bidders this strategy obtains a revenue equal to

$$\Rev_{STB}(\mathbf{v}) = \max_{(2)}\left(\left\{ \sum_{t=1}^{T} v_{i, t}  \middle|\; i \in [k] \right\}\right),$$

\noindent
where we write $\max_{(2)}(S)$ to denote the second largest element in $S$. We say that a mechanism (more specifically, a family of mechanisms, one for each $T$) is \textit{$\alpha$-competitive} with selling the business if for all $T$ and type profiles $\mathbf{v}$, 

$$\Rev(\mathbf{v}) \geq \alpha \Rev_{STB}(\mathbf{v}) - O_k(1).$$

It should be noted that $\Rev_{STB}$ is not the ``optimal'' policy for the seller when bidders are unconstrained, merely a natural choice. In fact, for any type profile $\mathbf{v}$, there is a mechanism that achieves revenue equal to $\Rev(\mathbf{v}) = \sum_{t} \max(\{v_{i, t} | i \in [k]\})$ on this specific type profile (even if bidders are limited-liability), namely the mechanism which hopes that the type profile is exactly $\mathbf{v}$ and posts prices accordingly. 

We will write $\vb_i$ to refer to the specific type of player $i$ (in the type profile $\vb$), and $|\vb_i| = \sum_{t=1}^{T} v_{i,t}$ to denote the total value of type $\vb_i$. For instance, using this notation, we can write $\Rev_{STB}(\mathbf{v}) = \max_{(2)}(\{|\vb_i|\, \mid\, i \in [k]\})$.

\section{Optimal mechanisms for selling to a single buyer}

% \subsection{Single buyer game}\label{sect:modelsing}

In the course of designing mechanisms and proving lower bounds for the multiple bidder game described in the previous section, it will be useful to think about a different game between a single seller and a single buyer. This game is identical to the $k=1$ variant of the multiple buyer game described in the previous section, with the following modifications:

\begin{itemize}
    \item The seller knows (at the beginning of the protocol) the total value $V$ of the buyer over all $T$ rounds.
    \item Instead of trying to compete with the benchmark of selling the business (the sum of the second-largest value per round), the seller wishes to be $\alpha$-competitive with respect to the total value $V$; i.e., the seller's mechanism should satisfy
    
    $$\Rev(\mathbf{v}) \geq \alpha V - O(1)$$
    
    for all $T$ and types $\mathbf{v}$.
\end{itemize}

We can define truthful, direct-revelation mechanisms as before. Since there is only a single buyer in this case, we drop all subscripts $i$ where relevant; e.g. a direct-revelation mechanism in the single buyer game is parametrized by sequences $r_{\mathbf{v}, t}$, $x_{\mathbf{v}, t}$, and $g_{\mathbf{v}, t}$.

In this section, we will show that there exists a $(1/e)$-competitive mechanism for the seller in the single buyer game. Moreover, we will show that this competitive ratio is optimal; no mechanism for the single buyer game can obtain more than  a $1/e$ fraction of the welfare. 

Our main technique for the both the upper and lower bound is understanding how the rate and allocation functions of the mechanism at time $t$ depend on the total amount of revenue received from the buyer so far. Our $(1/e)$-competitive mechanism will take the form of what we refer to as a ``pay-to-play'' mechanism, where the fraction of the item you are allocated at time $t$ is equal to some fixed function of the total amount you have paid to the seller so far.

\subsection{Continuous-time models}\label{sect:modelcont}

The analysis of the single buyer game is more clear in a continuous-time setting, where instead of being discretized into $T$ rounds, the buyer's value (and the various parameters of his mechanism) are described by well-behaved functions over a continuous interval of time. We formalize this continuous-time model here. The analogous upper and lower bounds in the discrete-time model are formulated and proved in Appendices \ref{sect:singlb-app} and \ref{sect:1p_disc_ub} -- they follow essentially the same logic, but are slightly messier.

In this continuous variant, our buyer has a Lipschitz-continuous value function $v(t):  [0, 1] \rightarrow [0, \infty)$ representing their value at time $t$. If the buyer has total value $V$ (known to the seller), this function $v(t)$ should satisfy $\int_{0}^{1}v(t)dt = V$. 

A direct-revelation mechanism in this setting specifies for each type $v(t)$ a continuous rate function $r_v(t): [0, 1] \rightarrow [0, 1]$ (the fraction of the infinitesimal item the buyer receives at time $t$) and a continuous payment function $x_v(t): [0, 1] \rightarrow [0, \infty)$. In order for this direct-revelation mechanism to be truthful, it must satisfy the following analogues of the corresponding constraints in Section \ref{sect:modelmulti}:

\begin{itemize}
    \item For all types $v$ and $t \in [0, 1]$,
    
    $$r_v(t)v(t) - x_v(t) \geq 0.$$
    
    \item For any types $v, v'$ and $\tau \in [0, 1]$ that satisfy

\begin{equation}
r_{v'}(t)v(t) - x_{v'}(t) \geq 0 \hspace{10 mm} \forall \;0 \leq t \leq \tau
\end{equation}

(i.e., it is possible for $v$ under limited liability constraints to imitate type $v'$ up until time $\tau$), it is the case that

\begin{equation}\label{eqn:dynamicic}
    \int_{0}^{1}(r_v(t)v(t) - x_v(t))dt \geq \int_{0}^{\tau}(r_{v'}(t)v(t) - x_{v'}(t))dt.
\end{equation}
\end{itemize}

Note that for continuous mechanisms, we don't need to allow for a final reimbursement; this reflects the fact that only an infinitesimal amount of the product is up for grabs at any point, so we can guarantee a buyer cannot gain positive utility by defecting.

The total revenue this mechanism achieves on type $v$ is given by $\int_{0}^{1} x_v(t)dt$; we denote this quantity as $X_v$. Our goal is to maximize our worst-case revenue over all types, $M = \inf_v X_v$; if a mechanism attains a specific $M$, we say it is $M$-competitive. 

% Note that in the continuous model, we no longer need to provide reimbursements at the end of the game; this is related to the fact that 

\paragraph{Reparametrization by welfare}

Define $w_v(t) = \int_{0}^{t}r_v(s)v(s)ds$ to be the total utility received by the bidder up until time $t$, not including payments (we call this quantity the \textit{welfare} of the bidder up until time $t$). Note that for $t \in [0, 1]$, $w_v(t)$ is monotone increasing from $0$ to $W_v = w_v(1)$. 

This allows us to reparametrize our existing functions (value, rate, payment) in terms of welfare. For example, we (abusing notation) let $r_v(w) = r_v(w_v^{-1}(w))$ to be the rate at the time when bidder $v$ has received welfare $w$. We similarly define $v(w)$ and $x_v(w)$.

\subsection{A $(1/e)$-competitive mechanism}\label{sect:singlb}

In this section we demonstrate a mechanism for the single buyer game which is $(1/e)$-competitive. As mentioned, this mechanism falls into the class of mechanisms that we call \textit{pay-to-play mechanisms}. 

A pay-to-play mechanism is specified by a continuous, weakly-increasing function $\rho(X): [0, 1] \rightarrow (0, 1]$, denoting the fraction of the item we allocate to the buyer if the buyer has given the seller a fraction $X$ of his welfare so far. The buyer is free to pay however much they wish to at time $t$, under the constraint that a limited-liability buyer cannot pay more than $\rho(X)$ of their value in any round. We claim that for any pay-to-play mechanism, the dominant strategy for any limited-liability buyer (regardless of their type $v$) is to pay as much as they can every round until they have paid an $X_{opt}$ fraction (for some value $X_{opt}$ depending only on $\rho$) of their total value $V$, and then never pay again. 

More formally, we have the following lemma.

\begin{lemma}\label{lem:paytoplay-cont}
Let $\rho(X): [0, 1] \rightarrow (0, 1]$ be a weakly-increasing continuous function and let $X_{opt}$ be a value of $x \in [0, 1]$ which maximizes the expression

$$\rho(x)\left(1 - \int_{0}^{x}\frac{dx'}{\rho(x')}\right).$$

Consider the pay-to-play mechanism defined by $\rho(X)$. Fix a type $v$ with total value $V$. Then there exists a dominant strategy for the buyer with type $v$ where they pay at least $X_{opt}V$ to the seller. Moreover, if $X_{opt}$ is a strict maximum, then this property holds for any dominant strategy.
\end{lemma}
\begin{proof}
Begin by fixing any dominant strategy for the buyer $x_{v}(t)$ (since there is only one player, this is just any strategy which optimizes the buyer's utility). We first note that if $\int_{t}^{1} x_{v}(t)dt > 0$ (the buyer pays some amount to the seller after time $t$), then $\int_{0}^{t}x_{\vb}(t)dt = \int_{0}^{t}r_{v}(t)v(t)dt = w_v(t)$ (the buyer passes along all reward at times $s < t$). In other words, the buyer's payments are front-loaded; i.e. any dominant strategy must pay as much as possible until it passes some threshold, then stop paying entirely. To see this, it suffices to note that since $\rho(X)$ is an increasing function of $X$, moving later payments earlier increases the value of all future rewards and hence is strictly optimal.

Any dominant strategy for the buyer can therefore be characterized by the total amount the buyer gives to the seller. We will show there is one dominant strategy where they give at least $X_{opt}V$. Assume that the buyer gives a total of $\overline{x}V$, for some $\overline{x} \in [0, 1]$. If the buyer pays the seller until time $\tau$, then $\tau$ must satisfy $\int_{0}^{\tau}x_v(t)dt = \overline{x}V$. But from the above argument, we know that $\int_{0}^{\tau}x_v(t)dt = w_v(\tau)$, so this means $\tau = w_v^{-1}(\overline{x}V)$.

Now, note that the eventual utility of the buyer is $\rho(\overline{x})(V - \int_{0}^{\tau}v(t)dt)$ (the buyer gets no net utility from the item until time $\tau$, and from then on they get a constant fraction $\rho(\overline{x})$ of the remainder of the item). 

Recall that $w_{v}(t) = \int_{0}^{t}r_{v}(s)v(s)ds$, so $\frac{dw_{v}}{dt} = r_{v}(t)v(t)$. This lets us reparametrize the integral $\int_{0}^{\tau}v(t)dt$ in terms of welfare via

\begin{eqnarray*}
\int_{0}^{\tau}v(t)dt &=& \int_{0}^{w_{v}(\tau)} v(t) \cdot \frac{1}{r_{v}(t)v(t)} dw \\
&=& \int_{0}^{\overline{x}V} \frac{dw}{r_{v}(w)} \\
&=& \int_{0}^{\overline{x}V} \frac{dw}{\rho(w/V)} \\
&=& V\int_{0}^{\overline{x}} \frac{dw'}{\rho(w')}.
\end{eqnarray*}

In the second-to-last equality, we have used the fact that while the buyer is passing along their full welfare to the seller, $r_{v}(w) = \rho(w/V)$. It follows that the net utility of the buyer is equal to

$$V\rho(\bar{x})\left(1 - \int_{0}^{\bar{x}}\frac{dw}{\rho(w)}\right).$$

Since $\overline{x}$ was chosen to maximize this expression, this is the maximal possible utility possible for the buyer and hence this is a dominant strategy. If $\overline{x}$ is a strict maximizer, any dominant strategy must pay a total of $\overline{x}V$ to the seller (or it will be dominated by this strategy). 
\end{proof}

% We reiterate that the reimbursements in Lemma \ref{lem:paytoplay} should be thought of as ``tie-breakers'' for the agent deciding between multiple comparable strategies. In most cases, they are negligible compared to the total value of the buyer; for a buyer with constant average value per round, their total value $V$ will be $\Theta(T)$, whereas this reimbursement is always at most $O(1)$. Contrast this with other dynamic mechanisms which, for example, refuse to pay the bidder anything until the end of the protocol (and thus break the spirit of limited liability).

With Lemma \ref{lem:paytoplay-cont} in hand, we can exhibit our $(1/e)$-competitive mechanism. 

\begin{theorem}\label{thm:mainalg}
There exists a mechanism for the seller which obtains $V/e$ total revenue.
\end{theorem}
\begin{proof}
Consider the pay-to-play mechanism defined by the function $\rho(w)$, where

$$\rho(w) = \begin{cases}
e^{ew - 1} & \mbox{ if } w \in [0, 1/e] \\
1 & \mbox{ if } w \in [1/e, 1]
\end{cases}$$

Note that for $x \leq 1/e$, we have that

\begin{eqnarray*}
\rho(w)\left(1 - \int_{0}^{x}\frac{dx'}{\rho(x')}\right) &=& e^{ex-1}\left(1 - \int_{0}^{x'}e^{-ex'+1}dx'\right) \\
&=& e^{ex - 1}e^{-ex} \\
&=& e^{-1}.
\end{eqnarray*}

For $x > 1/e$, this expression is decreasing (since $\rho(x)$ is constant for $x \in [1/e, 1]$, but $\int_{0}^{1} (1/\rho(x'))dx'$ is increasing). It follows that $X_{opt} = 1/e$ is a value of $x$ which maximizes utility. By Lemma \ref{lem:paytoplay-cont}, it follows that a pay-to-play mechanism with this choice of $\rho$ results in the seller receiving $V/e$ total revenue, as desired.
\end{proof}

\begin{remark}\label{remark:mech1}
Note that for the $\rho$ in Theorem \ref{thm:mainalg}, the expression in Lemma \ref{lem:paytoplay-cont} is only weakly maximized at $X_{opt} = 1/e$ (and in fact attains this maximum for all $X \in [0, 1/e]$). It is possible to perturb $\rho$ slightly so that the expression is \textit{strictly} maximized at some $X_{opt} \geq 1/e - \eps$ for some arbitrarily small $\eps$, thus satisfying the stronger conditions of Lemma \ref{lem:paytoplay-cont}. One option is to choose $\tilde{\rho}(w) = \rho(\min(w/(1-\eps), 1))$. This can be thought of as allocating the item via the original pay-to-play mechanism but as if the buyer actually had value $(1-\eps)V$ (to see why this works, see Remark \ref{remark:mech2}). Alternatively, it suffices to give any positive reimbursement at the end of the game contingent on the buyer paying at least $V/e$.
\end{remark}

\begin{remark}\label{remark:mech2}
The above analysis assumes that the total value of the buyer is exactly $V$, but in fact works as long as $V$ is any lower bound on the total value of the buyer (that is, it still guarantees a revenue of at least $V/e$). Indeed, if the buyer's true value is $V' \geq V$, then by stopping at a payment of $\bar{x}V$ they receive a net utility of

$$\rho(\overline{x})\left(V' - V\int_{0}^{\overline{x}}\frac{dw}{\rho(w)}\right).$$

For the $\rho$ in Theorem \ref{thm:mainalg}, this expression is strictly increasing for $\overline{x} \in [0, 1/e]$ (in particular, it equals $(V'-V)\rho(\overline{x}) + e^{-1}V$) and strictly decreasing for $\overline{x} > 1/e$, so it is is strictly dominant for the buyer to pay $V/e$.
\end{remark}

\begin{remark}\label{remark:mech3}
Likewise, although the above analysis assumes that the buyer is subject to limited liability and does not pay more than his value, this too is not necessary. It is still true that a non-limited liable agent will front-load payments as much as possible, and thus we will have that $\tau \leq w_{v}^{-1}(\overline{x}V)$ (instead of in the limited-liable case, where they are equal). The rest of the proof proceeds as before.

For example, in the extreme case where the buyer can pay as much as they want at the very beginning, they will simply pay $V/e$ at the very beginning and receive $(1 - 1/e)V$ total utility (in contrast to if they were limited-liable, in which case they receive $V/e$ utility).
\end{remark}

\begin{remark}
The discrete version of this mechanism is presented in Appendix \ref{sect:singlb-app}. As expected, it works essentially equivalently to the continuous mechanism of Theorem \ref{thm:mainalg}. Two small differences are: i) instead of guaranteeing a revenue of $V/e$, it only guarantees a revenue of $V/e - O(1)$ (this additive loss is inevitable due to the, and ii) it employs a small $O(1)$ reimbursement at the end of the protocol to guarantee truthfulness.
\end{remark}

\subsection{Upper bound of $1/e$}

In this section we show that no incentive compatible mechanism can extract more than $1/e$ of the welfare of the bidder. Throughout this section, we will assume without loss of generality that the total value $V$ of the bidder is normalized to $1$.

We will need the following auxiliary lemma.

\begin{lemma}\label{lem:hazardub}
Fix an $\alpha \in [0, 1]$, and let $f:[0,1] \rightarrow [0,1]$ be a function satisfying

$$\left(1 - \int_{0}^{x}\frac{dx'}{f(x')}\right)f(x) \leq \alpha$$

for all $x \in [0, 1]$. Then

$$\int_{0}^{\alpha \log (1/\alpha)} \frac{dx}{f(x)} \geq 1 - \alpha.$$
\end{lemma}
\begin{proof}
Define $g(x) = 1/f(x)$, so $g(x)$ satisfies

$$g(x) \geq \frac{1}{\alpha} - \frac{1}{\alpha}\left(\int_{0}^{x}g(x')dx'\right)$$

Let $\lambda(x) = \alpha e^{x/\alpha}$. Multiplying the above inequality by $\lambda(x)$ and integrating it from $0$ to $\alpha\log(1/\alpha)$, we get that

\begin{eqnarray*}
\int_{0}^{\alpha\log(1/\alpha)}\lambda(x)g(x)dx &\geq & \frac{1}{\alpha}\int_{0}^{\alpha\log(1/\alpha)}\lambda(x)dx - \frac{1}{\alpha}\int_{0}^{\alpha\log(1/\alpha)}\lambda(x)\int_{0}^{x}g(x')dx'dx \\
&=& \frac{1}{\alpha}\int_{0}^{\alpha\log(1/\alpha)}\lambda(x)dx - \frac{1}{\alpha}\int_{0}^{\alpha\log(1/\alpha)}g(x)\int_{x}^{\alpha\log(1/\alpha)}\lambda(x')dx'dx.
\end{eqnarray*}

From this, we have that

\begin{eqnarray*}
1 - \alpha &=& \frac{1}{\alpha} \int_{0}^{\alpha \log (1/\alpha)} \lambda(x)dx \\
&\leq & \int_{0}^{\alpha\log(1/\alpha)}\lambda(x)g(x)dx + \frac{1}{\alpha}\int_{0}^{\alpha\log(1/\alpha)}g(x)\int_{x}^{\alpha\log(1/\alpha)}\lambda(x')dx'dx \\
&=& \int_{0}^{\alpha\log(1/\alpha)}g(x)\left(\lambda(x)+ \frac{1}{\alpha}\int_{x}^{\alpha\log(1/\alpha)}\lambda(x')dx'\right)dx \\
&=& \int_{0}^{\alpha\log(1/\alpha)}g(x)dx,
\end{eqnarray*}

from which the desired result follows.
\end{proof}

We can now prove our main theorem.

\begin{theorem}\label{thm:singlb}
The worst-case revenue of any incentive compatible mechanism is at most $1/e$. 
\end{theorem}
\begin{proof}
In other words, we must show that for any mechanism $M \leq 1/e$.

For a type $v$ let $U_v = W_v - X_v$ be the total utility this type ends up with at the end of the game if they act truthfully. We will need the following lemma, which shows that if all $U_v$ are bounded by some constant and $M$ is large enough, we can bound $U_v$ by an even smaller constant. 

\begin{lemma}\label{lem:recurse}
If for all types $v$, $U_v \leq \alpha$ (for some $\alpha \in [0, 1]$), then for all types $v$, $U_v \leq f(\alpha) - M$, where

$$f(\alpha) = \alpha\log(1/\alpha) + \alpha$$
\end{lemma}
\begin{proof}
Fix any type $v$ and any $w \in [0, W_v]$; let $t_v(w) = w_v^{-1}(w)$. Fix a small $\eps > 0$, and consider a new type, which we will call $\spike(w, \eps)$ with value equal to $v(t)$ for $t \leq t_v(w)$, value equal to $0$ for $t \geq t_v(w) + \eps$, and a constant value $\nu_{\spike} = \frac{1}{\eps}\int_{t_v(w)}^{1}v(t)dt$ within the interval $[t_v(w), t_v(w) + \eps]$ (here $\nu_{\spike}$ is chosen so that the total value of this type equals $1$). 

This type can deviate as follows: it can pretend to be $v$ until time $t_v(w) + \eps$ and then abort. By performing this deviation, the type $\spike(w, \eps)$ receives utility at least the utility it gains during this spike, so

\begin{eqnarray*}
U'_{\spike(w, \eps)} &\geq& \nu_{\spike}\int_{t_v(w)}^{t_v(w) + \eps} r_v(t)dt - \int_{t_v(w)}^{t_v(w) + \eps} x_v(t)dt \\
&=& \left(\int_{t_v(w)}^{1}v(t)dt\right)\left(\frac{1}{\eps}\int_{t_v(w)}^{t_v(w) + \eps} r_v(t)dt\right) - \int_{t_v(w)}^{t_v(w) + \eps} x_v(t)dt.
\end{eqnarray*}

On the other hand, by dynamic incentive-compatibility it holds that the utility the spike receives if it plays truthfully, $U_{\spike(w, \eps)}$, is at least $U'_{\spike(w, \eps)}$. Since by assumption $U_{\spike(w, \eps)} \leq \alpha$, this means 

$$\alpha \geq \left(\int_{t_v(w)}^{1}v(t)dt\right)\left(\frac{1}{\eps}\int_{t_v(w)}^{t_v(w) + \eps} r_v(t)dt\right) - \int_{t_v(w)}^{t_v(w) + \eps} x_v(t)dt.$$

Since $r_v(t)$ is continuous, taking the limit as $\eps \rightarrow 0$, the term $\left(\frac{1}{\eps}\int_{t_v(w)}^{t_v(w) + \eps} r_v(t)dt\right)$ converges to $r_v(t_v(w)) = r_v(w)$, and the term $\int_{t_v(w)}^{t_v(w) + \eps} x_v(t)dt$ converges to $0$. It follows that

\begin{equation}\label{eqn:alphapre}
    \alpha \geq \left(\int_{t_v(w)}^{1}v(t)dt\right)r_v(w).
\end{equation}

We'll now express the integral in (\ref{eqn:alphapre}) in terms of $w$. First, note that since $\int_{0}^{1} v(t)dt = 1$, 

$$\int_{t_v(w)}^{1}v(t)dt = 1 - \int_{0}^{t_v(w)}v(t)dt.$$

We will perform a substitution to express $\int_{0}^{t_v(w)}v(t)dt$ as an integral in terms of $w$. Since $w_{v}(t) = \int_{0}^{t}r_v(s)v(s)ds$, $\frac{dw_v(t)}{dt} = r_v(t)v(t)$ and $dt = \frac{dw_v(t)}{r_v(t)v(t)}$. It follows then that 

\begin{equation}\label{eqn:ttow}
    \int_{0}^{t_v(w)}v(t)dt = \int_{0}^{w}\frac{dw'}{r_v(w')}.
\end{equation}

Subsituting this into (\ref{eqn:alphapre}), we have that:

\begin{equation}\label{eqn:alpha}
    \alpha \geq \left(1 - \int_{0}^{w}\frac{dw'}{r_v(w')}\right)r_v(w).
\end{equation}

Note that by Lemma \ref{lem:hazardub}, (\ref{eqn:alpha}) implies that 

\begin{equation}\label{eqn:hazardlem}
    \int_{0}^{\alpha\log(1/\alpha)} \frac{dw}{r_v(w)} \geq 1 - \alpha.
\end{equation}

Now, since $\int_{0}^{1}v(t)dt = 1$ and $w(1) = W_v$, substituting $w = W_v$ into (\ref{eqn:ttow}) we have that

\begin{equation}\label{eqn:wmeasure}
    \int_{0}^{W_v} \frac{1}{r_v(w)} dw = 1.
\end{equation}

    If $W_v > \alpha\log(1/\alpha)$, we therefore have that 
    
    \begin{eqnarray*}
    1 &=& \int_{0}^{W_v} \frac{1}{r_v(w)} dw \\
    &=& \int_{0}^{\alpha\log(1/\alpha)} \frac{1}{r_v(w)} dw + \int_{\alpha\log(1/\alpha)}^{W_v} \frac{1}{r_v(w)} dw\\
    &\geq & (1 - \alpha) + (W_v - \alpha\log(1/\alpha)).
    \end{eqnarray*}
    
    where here we've applied (\ref{eqn:hazardlem}) and the fact that $r_v(w) \leq 1$. For this to be true, we must have
    
    \begin{equation}\label{eqn:wbound}
        W_v \leq \alpha \log(1/\alpha) + \alpha.
    \end{equation}
    
    On the other hand, if $W_v \leq \alpha\log(1/\alpha)$, $W_v$ also satisfies (\ref{eqn:wbound}). Therefore in either case, $W_v$ is at most $\alpha \log(1/\alpha) + \alpha = f(\alpha)$.
    
    Since $U_v = W_v - X_v$, and since $X_v \geq M$ for all types $v$ (by definition of $M$), it follows that $U_v \leq f(\alpha) - M$, as desired.
\end{proof}

Now, note that if $f(\alpha) - M < \alpha$ for all $\alpha \in [0, 1]$, repeatedly applying Lemma \ref{lem:recurse} (starting from $\alpha = 1$) will eventually imply $U_{v} < 0$ for all types $v$. But this is clearly not possible in any incentive compatible mechanism (any type could improve their situation by immediately deviating and not paying). 

We claim that $f(\alpha) - M < \alpha$ for all $\alpha \in [0, 1]$ if $M > 1/e$, thus completing the proof. One way to see this is to show that the maximum value of $f(\alpha) - \alpha$ on the interval $[0, 1]$ is at most $1/e$. But $f(\alpha) - \alpha = \alpha\log(1/\alpha)$, which is maximized at $\alpha = 1/e$ and has a maximum value of $1/e$, as desired.

\end{proof}

\section{Selling to multiple bidders}

In this section, we apply the results from the single buyer problem to the problem of selling to multiple bidders. 

We begin (in Section \ref{sect:2bid}) by considering the case when there are only two bidders. We show that even with just two bidders, it is impossible to get arbitrarily close to the benchmark of selling the business (in particular, we show that no algorithm can achieve a better competitive ratio than $2/e$). On the other hand we show that a constant competitive-ratio is achievable: by adapting the pay-to-play mechanism from the single buyer game, we can get a $(1/2e)$-competitive mechanism. 

For $k$ bidders, this same strategy results in a $(1/ke)$-competitive mechanism. We show how to improve this by adding competition in the form of a first-price auction. The resulting mechanism is no longer truthful, but has the guarantee that it is $\Theta(1)$-competitive as long as all bidders play non-dominated strategies.

\subsection{Two bidders}\label{sect:2bid}

We begin by showing that there is no limited liability mechanism which is guaranteed to get more than a factor of $2/e$ of the benchmark of selling the business. 

\begin{theorem}\label{thm:2plb}
There is no truthful mechanism for two bidders which is $\alpha$-competitive against selling the business for $\alpha > 2/e$.
\end{theorem}
\begin{proof}[Proof Sketch]
Fix a $\nu > 0$ and a sufficiently large $T$. Let $V = \nu T$. From our upper bound (Theorem \ref{thm:singlb}), we know that for any truthful mechanism for the single bidder game, there exists a value profile $\vb$ with $\sum_t v_t = V$ such that our mechanism receives revenue at most $V/e$ on this profile. By the minimax theorem, this means there exists a distribution $\mu$ over value profiles $\vb$ with $|\vb| = V$ such that no mechanism receives expected revenue more than $V/e$ on a value profile sampled from this distribution. (There is a technical issue here in that von Neumann's minimax theorem requires finite action spaces, but here there are an uncountably infinite set of valuations and protocols. We address this by performing appropriate discretizations -- the full detailed proof is in Appendix \ref{sect:2p-disc-ub}). 

Consider the following distribution $\D$ of instances for the two bidder game. All instances will have $2T$ rounds. The valuation profile $\vb_1$ of the first bidder will be sampled so that $(v_{1, 1},  \dots, v_{1, T})$ is sampled according to $\mu$ and $v_{1, t} = 0$ for $t > T$. Likewise, $\vb_2$ will be sampled so that $(v_{2,T+1}, \dots, v_{2, 2T})$ is (independently from $\vb_1$) sampled according to $\mu$ and so that $v_{2, t} = 0$ for $t \leq T$. 

Assume to the contrary that there exists a truthful mechanism $\mathcal{M}$ for two bidders which is $\alpha$-competitive for some $\alpha > 2/e$. Since $|\vb_1| = |\vb_2| = V$ for all instances $\vb$ in the support of $\D$, $\Rev_{STB}(\vb) = V$, and this means that this mechanism receives expected revenue at least $\alpha V$ over instances sampled from $\D$. 

Now, note that from the mechanism $\mathcal{M}$ we can construct two mechanisms $\mathcal{M}_1$ and $\mathcal{M}_2$ for the single-bidder game. To construct $\mathcal{M}_1$, after soliciting the bidder's valuation $\vb$, additionally sample a dummy valuation $\vb'$ from $\mu$. Consider running $\mathcal{M}$ on the valuation profile $(\vb, \vb')$: since $\vb$ and $\vb'$ have disjoint supports, $\mathcal{M}$ chooses some sequence of allocations / payments for the first bidder for the first $T$ rounds and some sequence of allocations / payments for the second bidder for the second $T$ rounds. Use this sequence of allocations / payments for the first $T$ rounds as the allocations / payments for $\mathcal{M}_1$; since $\mathcal{M}$ is truthful, it follows that $\mathcal{M}_1$ is also truthful. We construct $\mathcal{M}_2$ symmetrically.

By construction, the expected revenue of $\mathcal{M}$ on instances drawn from $\D$ is the same as sum of the expected revenues of mechanisms $\mathcal{M}_1$ and $\mathcal{M}_2$ on instances drawn from $\mu$. Since $\mathcal{M}$ receives expected revenue at least $\alpha V$, this means that either $\mathcal{M}_1$ or $\mathcal{M}_2$ must receive expected revenue at least $(\alpha / 2) V$ on instances from $\mu$. But $(\alpha / 2)V > V/e$, contradicting the fact that no mechanism can receive expected revenue larger than $V/e$ from $\mu$. The theorem statement follows.
\end{proof}

Note that since we can always embed the 2 bidder game in the $k$ bidder game for any $k > 2$, we have the immediate corollary that there is no $\alpha$-competitive mechanism for $k$ bidders with $\alpha \geq 2/e$. 

\begin{corollary}\label{cor:kbidderub}
For any $\eps > 0$ and integer $k \geq 2$, there is no mechanism for $k$ $\eps$-truthful bidders which is $\alpha$-competitive against selling the business for $\alpha > 2/e$.
\end{corollary}

\begin{proof}
See Appendix \ref{sect:kp-disc-ub}.
\end{proof}

On the other hand, we can reuse the pay-to-play mechanism of Section \ref{sect:singlb} to get a $(1/2e)$-competitive mechanism. This mechanism works as follows:

\begin{enumerate}
    \item At the beginning of the game, ask both players for their total values $V_1 = \sum_{t} v_{1, t}$ and $V_2 = \sum_{t} v_{2, t}$.
    \item Define 
    
    $$\rho(w) = \begin{cases}
e^{ew - 1} & \mbox{ if } w \in [0, 1/e] \\
1 & \mbox{ if } w \in [1/e, \infty)
\end{cases}$$

    (Recall that this is the rate function for the single buyer mechanism in Theorem \ref{thm:mainalg}).
    
    \item
    In a given round, let $X_1$ and $X_2$ be the total payments of players 1 and 2 thus far. Allocate fraction $\rho(2X_1/V_2)/2$ of the item to player 1, and allocate fraction $\rho(2X_2/V_1)/2$ of the item to player 2. Note that since $\rho(x) \leq 1$, our allocation is valid.
\end{enumerate}

\begin{theorem}\label{thm:twopmech}
This mechanism for 2 bidders is $(1/2e)$-competitive mechanism against selling the business.
\end{theorem}
\begin{proof}
Assume without loss of generality that $V_1 \geq V_2$. We therefore must show that this mechanism obtains revenue $V_2/(2e)$.

From the perspective of player 1, they are playing the single buyer mechanism for half of the item each round (for which their total valuation is $V_1/2$) with the lower bound on their value of $V_2 / 2$. Since $V_1 / 2 \geq V_2 / 2$, the single buyer mechanism guarantees that player 1 will pay at least $(1/e) \cdot (V_2 / 2) = V_2/(2e)$, as desired.
\end{proof}

\subsection{More than two bidders}\label{sect:kbidders}

What approximation factor can we achieve when we have $k > 2$ bidders? Adapting the mechanism of Theorem \ref{thm:twopmech} (by splitting the item into $k$ pieces and running the single buyer mechanism on each piece with the highest other reported value), we can obtain an approximation factor of $1/ek$. In this section, we will show that it is in fact possible to achieve a constant approximation factor independent of $k$. Unlike the other mechanisms in this paper, this mechanism is not a direct revelation mechanism -- nonetheless, we will show that as long as bidders follow non-dominated strategies, we will receive expected revenue within a constant factor of $\Rev_{STB}(\vb)$. 

\subsubsection{Selling shares via a first-price auction}

Assume, to begin, that we know any lower bound $V^*$ on the largest total value of any bidder. We will first show how to construct a mechanism that guarantees we obtain revenue at least a constant fraction of $V^*$. This mechanism will also have the property that it requires reimbursements of size $O(T)$ at the end of the protocol. Note that this is not inherently at odds with limited liability: our protocol will still never expect a bidder to pay more than their value in any given round. Nevertheless, we will show how to remove this constraint in Section \ref{sect:reimbursement}.

Roughly, our mechanism proceeds as follows. Each round we will split the item into two halves, which we will allocate in different ways. One half of the item we will allocate via a first-price auction among all $k$ bidders (each bidder will submit one bid per round in advance). We will allocate a fraction $\rho(X_i) = 10(X_i/V^*)$ of the other half of the item (i.e. $\rho(X_i)/2$ of the item) to bidder $i$, where $X_i$ is the total payment of bidder $i$ up to this round in the first-price auction.

Now, note that it is possible for $\sum_{i} \rho(X_i) > 1$; this would make it impossible to allocate the second half of the item as described above. To get around this, as soon as $\sum_{i} \rho(X_i) = 1$ we will freeze the allocation rates as is. That is, we will continue allocating the first half of the item to the bidder with the highest bid via a first-price auction, but stop updating the values $X_i$. Finally, at the end of the protocol, we will reimburse each bidder for the amount they paid in the first-price auction after the allocation rates have been fixed.

More formally, our mechanism operates as follows:

\begin{enumerate}
    \item At the beginning of the mechanism, the seller asks each bidder $i$ to (simultaneously) report their desired bid $b_{i, t}$ in the first price auction at time $t$. 
    \item In round $t$, the seller splits the item into two equal halves which they allocate separately. We will keep track of a value $X_i(t)$ for each bidder $i$ representing the amount they have paid the seller up until round $t$ before the allocation cap is hit. A fraction $\rho(X_i(t)) = 10(X_i(t)/V^*)$ of the first half (the \textit{allocation half}) is allocated to bidder $i$. We will guarantee $\sum_i \rho(X_i(t)) \leq 1$ for all $t$. 
    \item The second half of the item (the \textit{auction half}) is allocated to the bidder $i = \arg\max_{i} b_{i, t}$ (if multiple bidders share the same highest bid, split it evenly among them). This bidder is expected to pay $b_{i, t}$ this round (if they do not, we refuse to ever allocate to them in the future). If this causes $\sum_{i} \rho(X_i(t))$ to exceed $1$, decrease the payment until it exactly equals $1$. Note that if $\sum_{i} \rho(X_i(t)) = 1$, the payment in this round does not contribute towards $X_i(t)$.
    \item Finally, at the end of the protocol, as long as they successfully paid their bids when they were chosen in step (3), we provide each bidder a reimbursement equal to the amount they paid the seller after the allocation cap was hit (i.e. after the first round where $\sum_{i} \rho(X_i(t)) = 1$). In addition, provide each bidder any $O(1)$ reimbursement $g > 1$ in case the other reimbursement is $0$. This makes it dominant for each bidder to only report bids $b_{i, t}$ which they have the ability of paying (i.e. $b_{i, t}$ satisfying $b_{i, t} \leq v_{i, t}/2$) and also to participate for the entire protocol. 
\end{enumerate}

Note that (for this mechanism) a pure strategy for bidder $i$ consists of their choice of $b_{i, t}$ to report at the beginning of the mechanism, along with a rule for when to defect from the mechanism. Recall that we write $s_i$ to denote a pure strategy for bidder $i$, $s_{-i}$ to denote a profile of pure strategies for the other $k-1$ bidders, and $U_i(s_{i}, s_{-i})$ the utility received by bidder $i$ when they play strategy $s_{i}$ and the other bidders play $s_{-i}$. We call a pure strategy $s_{i}$ weakly-dominated if there exists a (possibly mixed) strategy $\sigma'_{i}$ such that $U_{i}(s_{i}, s_{-i}) \leq U_{i}(\sigma'_{i}, s_{-i})$ for all $s_{-i}$, and if there exists at least one choice of $s_{-i}$ such that $U_{i}(s_{i}, s_{-i}) < U_{i}(\sigma'_{i}, s_{-i})$; a strategy is non-dominated if it is not weakly-dominated (and a mixed strategy is non-dominated if it is supported on non-dominated strategies). 

For example, note that since we offer a reimbursement of at least $g > 1$ at the end of the protocol, any strategy which defects (doesn't pay $b_{i, t}$ when they are expected to) is dominated; by not paying $b_{i, t}$ they possibly gain $b_{i, t}$ but lose $g$ (plus any other utility they may have obtained) for a net negative loss. In other words, all non-dominated strategies will never defect.

We show that, as long as all bidders are playing non-dominated strategies, the above mechanism extracts a constant fraction of $V^*$. 

\begin{theorem}\label{thm:kbiddermech-fpa}
Assume we are given a lower bound $V^*$ on the largest total value belonging to any bidder. Then, as long as all bidders play non-dominated strategies, the above mechanism achieves revenue at least $\alpha V^* - O(1)$ for some constant $\alpha > 0$ (independent of the number of bidders $k$).
\end{theorem}

The full proof of Theorem \ref{thm:kbiddermech-fpa} can be found in Appendix \ref{sect:kp-disc-lb}. The main intuition behind the proof is as follows. Without loss of generality, assume that bidder $1$ is the bidder with the highest total value (and so in particular, $V_1 \geq V^*$). We argue that for ``sufficiently early'' rounds $t$, if bidder $1$ ever bids less than $v_{1, t}/2$ (their value for the auction half of the item this round), then this strategy is dominated by the strategy where they bid $v_{1, t}/2$. This follows from the construction of the mechanism: increasing their bid this round costs bidder $1$ a little bit more if they win, but this is more than paid back by the value they obtain from the allocation half of the item over the remaining rounds. In particular, we can think of this as the bidder paying one unit of cost to obtain a $5/V^*$ share of the remaining value of the item; if the bidder values the remaining rounds at $V^*/5$ or more, the bidder should be happy to pay this amount (and in fact, the bidders in the mechanism need to compete for the ability to buy these shares in a first-price auction).

Since $V_1 \geq V^*$, there are many rounds where the bidder values the remainder of the item (i.e. the total value of the item across the remaining rounds) at over $V^*/5$; in fact, this happens for at least $4V_1/5$ of the bidder's welfare. If the auction is active for all these rounds, then it generates revenue at least $2V_1/5$ if bidder $1$ is playing a non-dominated strategy (and is thus bidding at least half their value per round). But if the auction ever stops, this means that $\sum_{i} \rho(X_i) = 1$, and therefore that $\sum_{i} X_i \geq V^*/10$, so we have received $V^*/10$ in revenue. In either case, we have received at least a constant fraction of $V^*$ (since $2V_1/5 \geq 2V^*/5$).

Why do we reimburse each bidder for the amount they paid after the allocation rates have been fixed? Note that if we do not do this, then bidders may be incentivized to shade their bids for rounds where they think the cap has been hit, and we can no longer argue that they will bid their full value. Providing this reimbursement allows us to encourage bidders to bid their full value for the item and incentivizes them to hit the cap.

\begin{remark}
In fact, Theorem \ref{thm:kbiddermech-fpa} does not require all $k$ bidders to play non-dominated strategies -- it merely requires the bidder with largest total value to play a non-dominated strategy. 

Similarly, although we assume bidders cannot bid above their true value, note that the mechanism in Theorem \ref{thm:kbiddermech-fpa} is still $\alpha$-competitive even if some bidders can bid above their true value: it is still in a bidder's interest to bid at least their value in early rounds of the first price auction (and in fact, it might be in their interest to bid above their value, if they can afford to do so). 
\end{remark}

\begin{remark}
This mechanism has the additional caveat that, unlike our previous mechanisms (and the $O(1/k)$-competitive generalization), it does require the bidders to know their individual per round values in advance (at least for Theorem \ref{thm:kbiddermech-fpa} to hold). Since bidders' bids can affect the utility of other bidders, allowing bidders to bid dynamically can lead to undominated grim trigger strategies where bidders punish each other for not bidding according to a specific schedule. It is an interesting open question whether it is possible to adapt this mechanism to remove these constraints. 
\end{remark}

\subsubsection{Soliciting bidders' values}

The mechanism of Theorem \ref{thm:kbiddermech-fpa} assumes we have a lower bound $V^*$ on the highest total value of a bidder, and extracts a constant fraction of this bound as revenue. However, our ultimate goal is to construct a completely prior-free mechanism that extracts constant fraction of the second-highest total value $\Rev_{STB}(\vb)$.

To do so, we will at the very beginning of the mechanism solicit the value of $\Rev_{STB}(\vb)$ using a similar technique as in our construction of the $k=2$ algorithm. Specifically, we will do the following at the very beginning of the protocol:

\begin{itemize}
    \item At the beginning of the mechanism, the seller asks each bidder $i$ to (simultaneously) report their total value $V_i$. (This should be done simultaneously with the reporting of bids in the mechanism of Theorem \ref{thm:kbiddermech-fpa}). 
    \item The seller splits the bidders randomly into two subsets $S_{price}$ and $S_{play}$. Each bidder is independently assigned to $S_{price}$ with probability $1/2$ and to $S_{play}$ with probability $1/2$. (If either set is empty the seller immediately stops the auction, never allocating the item).
    \item Define $V^* = \max_{i \in S_{price}} V_i$. The seller now removes all players in $S_{price}$ from the auction (i.e., the seller will only allocate to players in $S_{play}$ in the future).
    \item Run the mechanism from Theorem \ref{thm:kbiddermech-fpa} on the bidders in $S_{play}$ with this $V^*$.
\end{itemize}

Note that $1/4$ of the time, the bidder with the highest value belongs to $S_{play}$ and the bidder with the second highest value belongs to $S_{price}$. In this case, $V^* = \Rev_{STB}(\vb)$ and the conditions of Theorem \ref{thm:kbiddermech-fpa} are satisfied with constant probability (assuming bidders report their values truthfully). Moreover, note that the value $V$ that a bidder reports in this step is completely independent from their eventual utility (in particular, if they are selected to belong to $S_{play}$, their value $V$ is never used). It is therefore never better for a bidder to report a value $V'$ not equal to their true value $V$. 

If bidders cannot pay more than their value each round, it is possible to make it dominant for each bidder to report their true value $V$. To do so, with some small probability $\delta > 0$, instead of running the rest of the mechanism after step (1), run the following procedure: allocate $\eps \approx 1/kT$ of the item to each buyer each round. If buyer $i$ has paid a total of $\eps V_i$ at the end of the mechanism, reimburse them an additional $2\eps V_i$.

If we do this, then bidder $i$ receives a total utility of $\eps V_i - \eps V'_i + 2\eps V'_i = \eps (V_i + V'_i)$ by misreporting $V'_i$ instead of one's true valuation $V_i$; to maximize this, a bidder will want to report the largest $V'_i$ they can possibly afford, which is $V_i$ (they cannot report a larger $V'_i$ while paying at most their value every round). It is therefore dominated for a bidder $i$ to report any value $V'_i \neq V_i$.

Our main theorem for selling to $k$ buyers therefore follows.

\begin{theorem}\label{thm:kbiddermech}
As long as all bidders play non-dominated strategies, the above mechanism is an $\alpha$-competitive mechanism for $k$ bidders for some constant $\alpha > 0$ (independent of $k$).
\end{theorem}

\subsubsection{Removing large reimbursements}\label{sect:reimbursement}

As written, Theorem \ref{thm:kbiddermech} could require post-protocol reimbursements of up to $O(T)$, since we reimburse bidders for all payments after the allocation rates for the allocation half have been fixed. Although this is not at odds with limited-liability, ideally we could reduce these reimbursements to $O(1)$ as in our other protocols.

Since we reimburse bidders for their payments after the allocation cap has been hit, one simple idea is to simply remove these payments. In other words, after the allocation cap has been hit, still allocate the item to the bidder with the highest bid, but charge them nothing. This largely works, but has the unfortunate side effect that now bidders might gamble and submit bids larger than their value for later rounds where they think the allocation cap has already been hit. Note that this was originally not an issue because in the first-price auction a bidder had to immediately pay their bid if they won the item, thus checking that they in fact could afford it. But now, it might be strategic for a bidder to report a higher bid than their value for a late round hoping that the cap has been met (and if the cap has not been met, defecting at that point).

Nonetheless, we can still address this in the following (admittedly contrived) way. The key idea is to spot check that bidders are capable of paying the bids they have reported. More specifically, for some uniformly randomly selected constant fraction (e.g. $50\%$) of the rounds, instead of running the mechanism described above, allocate a tiny $\eps$ of the item (for some $\eps \approx 1/kT$) to each bidder and ask them to pay $2\eps$ times their bid (since their bid was for half the item, this should be at most $\eps$ times their value). If they fail to do this, eject them from the protocol. 

We then can show that with this change, it is dominated for a limited-liability bidder to bid higher than their value. In particular, the expected increase in utility they can possibly get from bidding above their value is outweighed by the expected loss in future utility (including the $O(1)$ reimbursement at the end) which they would lose from being forced to defect.

\begin{corollary}\label{cor:small_reimbursements}
It is possible to implement Theorem \ref{thm:kbiddermech} with reimbursements of size $O(1)$. 
\end{corollary}

Details and the full proof of Corollary \ref{cor:small_reimbursements} can be found in Appendix \ref{sect:kp-disc-lb}.

\bibliographystyle{ACM-Reference-Format}
\bibliography{references}

\appendix

\section{Omitted proofs}

\subsection{Mechanism for one buyer in the discrete model}\label{sect:singlb-app}

In this appendix, we extend the mechanism of Theorem \ref{thm:mainalg} (showing it is possible to obtain revenue $V/e$ in the single-buyer game) to the discrete model.

As in the continuous setting, a pay-to-play mechanism is specified by a continuous, weakly-increasing function $\rho(X): [0, 1] \rightarrow (0, 1]$, denoting the fraction of the item we allocate to the buyer if the buyer has given the seller a fraction $X$ of his welfare so far. The buyer is free to pay however much they wish to at time $t$, under the constraint that a limited-liability buyer cannot pay more than $\rho(X)$ of their value in any round. In our discrete setting, a pay-to-play mechanism may additionally provide a small ($O(1)$) reimbursement to the buyer at the end of the game if they pay at least some specified fraction of their welfare to the seller. 

We have the following discrete analogue of Lemma \ref{lem:paytoplay-cont}.

\begin{lemma}\label{lem:paytoplay-disc}
Let $\rho(X): [0, 1] \rightarrow (0, 1]$ be a weakly-increasing Lipschitz-continuous function and let $X_{opt}$ be a value of $x \in [0, 1]$ which maximizes the expression

$$\rho(x)\left(1 - \int_{0}^{x}\frac{dx'}{\rho(x')}\right).$$

Consider the pay-to-play mechanism defined by $\rho(X)$ where the buyer with type $\vb$ receives a reimbursement of $\Theta(1)$ if they pay at least $X_{opt}V$. Fix a type $\vb$ with total value $V \geq 1$. Then there exists a dominant strategy for the buyer with type $\vb$ where they pay at least $X_{opt}V$ to the seller. 
\end{lemma}
\begin{proof}
Begin by fixing any dominant strategy for the buyer $(r_{\vb, t}, x_{\vb, t})$ (since there is only one player, this is just any strategy which optimizes the buyer's utility). We first note that if $x_{\vb, t} > 0$, then $x_s = r_sv_s$ for all $s < t$ (the bidder passes along all reward at times $s < t$). In other words, the buyer's payments are front-loaded; i.e. any dominant strategy must pay as much as possible until it passes some threshold, then stop paying entirely. To see this, it suffices to note that since $\rho(w)$ is an increasing function of $w$, moving later payments earlier increases the value of all future rewards and hence is strictly optimal.

Any dominant strategy for the buyer can therefore be characterized by the total amount the buyer gives to the seller. We will show there is one dominant strategy where they give at least $X_{opt}V$. Assume that the buyer gives a total of $\overline{x}V$, for some $\overline{x} \in [0, 1]$. We introduce the following notation. 

Let $w_t = \frac{1}{V}\sum_{s \leq t} r_{\vb, s}v_{\vb, s}$ be the (normalized by $V$) total amount the buyer receives up until time $t$. We will divide time (specifically, the time where the buyer pays the seller) into equal regions by $w_t$. In particular, let $\Delta = 1/V$ and $N = \lceil \overline{x} V \rceil$, and set $T_i = \{ t | w_t \in [i\Delta, (i+1)\Delta)\}$ for $0 \leq i < N$. Let $s_i = \min T_i$ and $e_i = \max T_i$ (so $T_i = [s_i, e_i]$).

Let $V_i = \sum_{t \in T_i} v_t$ and $W_i = \sum_{t \in T_i} x_t = \sum_{t \in T_i} r_tv_t$. Let $U$ be the total utility of the bidder at the end of the game (but before any reimbursement). Note that we can approximate $U$ (to within $O(1)$) via

\begin{equation}\label{eqn:bidder_util}
\left|U - \rho(\overline{x})\left(V - \sum_{i=0}^{N-1}V_i\right)\right| \leq 1,
\end{equation}

\noindent
since the bidder receives all their reward after they stop paying at a rate of $\rho(\overline{x})$ (the gap comes from not paying the full $r_tv_t$ on their last turn, and is at most $1$). We now claim that 

\begin{equation}\label{eqn:abs_bound}
\left|\sum_{i=0}^{N-1} V_i - V \int_{0}^{\overline{x}}\frac{1}{\rho(w)}dw \right| \leq \frac{2}{\rho(0)}.
\end{equation}

To see this, note that since for $t \in T_i$, $r_t \in [\rho(w_{s_i}), \rho(w_{e_i})]$, it follows that $W_i \in [ \rho(w_{s_i})V_i, \rho(w_{e_i})V_i]$, and therefore that $V_i \in [W_i/\rho(w_{e_i}), W_i/\rho(w_{s_i})]$. We can therefore write 

$$\sum_{i=0}^{N-1}V_i \leq \sum_{i=0}^{N-1}\frac{W_i}{\rho(w_{s_i})} $$

Now, note that $W_i = V(w_{s_{i+1}} - w_{s_i})$, so we can interpret this as a Riemann sum for $\int_{0}^{\overline{x}}1/\rho(w)dw$. The error between this sum and the integral is bounded by the range of $1/\rho(w)$ (which is $1/\rho(0)$) multiplied by the maximum width of an interval. By the construction of $T_i$, $w_{s_{i+1}} - w_{s_i} \leq \Delta + \frac{1}{V}$. It follows that

\begin{eqnarray*}
\sum_{i=0}^{N-1}V_i &\leq & \sum_{i=0}^{N-1}\frac{W_i}{\rho(w_{s_i})} \\
&= & V\sum_{i=0}^{N-1}\frac{(w_{s_{i+1}} - w_{s_i})}{\rho(w_{s_i})} \\
&\leq & V\left(\int_{0}^{\overline{x}}\frac{1}{\rho(w)}dw + \frac{1}{\rho(0)}\left(\Delta + \frac{1}{V}\right)\right) \\
&\leq & V\left(\int_{0}^{\overline{x}}\frac{1}{\rho(w)}dw\right) + \frac{2}{\rho(0)}.
\end{eqnarray*}

Similarly (by using the fact that $V_i \geq W_i/\rho(w_{e_i})$), we can show that $\int_{0}^{\overline{x}}\frac{1}{\rho(w)}dw - \sum_{i=0}^{N-1}V_i \leq 2/\rho(0)$, from which (\ref{eqn:abs_bound}) follows. 

Combining equations (\ref{eqn:bidder_util}) and (\ref{eqn:abs_bound}), it follows that

\begin{equation}
    \left| U - \rho(\overline{x})\left(1 - \int_{0}^{\overline{x}}\frac{1}{\rho(w)}dw\right)V\right| \leq 1 + \frac{2\rho(\overline{x})}{\rho(0)} = O(1).
\end{equation}

Since $\rho(\overline{x})\left(1 - \int_{0}^{\overline{x}}\frac{1}{\rho(w)}dw\right)$ is maximized for $\overline{x} = X_{opt}$, if the seller offers an $O(1)$ reimbursement (of size at least $1 + 2/\rho(0)$) to the buyer as long as they pay at least $X_{opt}V$, it follows that it is a dominant strategy for the buyer to pay at least $X_{opt}V$, as desired.
\end{proof}

We reiterate that the reimbursements in Lemma \ref{lem:paytoplay-disc} should be thought of as ``tie-breakers'' for the agent deciding between multiple comparable strategies. In most cases, they are negligible compared to the total value of the buyer; for a buyer with constant average value per round, their total value $V$ will be $\Theta(T)$, whereas this reimbursement is always at most $O(1)$. Contrast this with other dynamic mechanisms which, for example, refuse to pay the bidder anything until the end of the protocol (and thus break the spirit of limited liability).

We can now prove the discrete-time analogue of Theorem \ref{thm:mainalg}.

\begin{theorem}\label{thm:mainalg-disc}
There exists a mechanism for the seller which obtains $V/e - O(1)$ total revenue.
\end{theorem}
\begin{proof}
Consider the following mechanism for the seller: at each round, if the bidder has given $wV$ in total so far to the seller, allocate fraction $\rho(w)$ of the item to the bidder, where

$$\rho(w) = \begin{cases}
e^{ew - 1} & \mbox{ if } w \in [0, 1/e] \\
1 & \mbox{ if } w \in [1/e, 1]
\end{cases}$$

%In addition, give the bidder a final payment of $1$ if they have given a total of $V/e$ by the end of the game. (This final payment is not necessary, but will help the analysis below -- it is easy to see that removing it can only decrease the amount the bidder will give the seller by at most $1$). 

Note that for $x \leq 1/e$, we have that

\begin{eqnarray*}
\rho(x)\left(1 - \int_{0}^{x}\frac{dx'}{\rho(x')}\right) &=& e^{ex-1}\left(1 - \int_{0}^{x'}e^{-ex'+1}dx'\right) \\
&=& e^{ex - 1}e^{-ex} \\
&=& e^{-1}.
\end{eqnarray*}

For $x > 1/e$, this expression is decreasing (since $\rho(x)$ is constant for $x \in [1/e, 1]$, but $\int_{0}^{1} (1/\rho(x'))dx'$ is increasing). It follows that $X_{opt} = 1/e$ is a value of $x$ which maximizes utility. By Lemma \ref{lem:paytoplay-disc}, it follows that a pay-to-play mechanism with this choice of $\rho$ results in the seller receiving $V/e - O(1)$ total revenue, as desired.

\end{proof}

\subsection{Upper bound for one buyer in the discrete model}\label{sect:1p_disc_ub}

In this appendix, we extend the proof of Theorem \ref{thm:singlb} (proving an upper bound for the single buyer game) to the discrete model. Specifically, we will show the following theorem.

\begin{theorem}\label{thm:disc_ub}
For any $c > \frac{1}{e}$ and $d \geq 0$, there exists a sufficiently large $T$ such that no mechanism for the single buyer game satisfies

$$\Rev(\vb) \geq c V + d$$

\noindent
for all types $\vb$.
\end{theorem}

Before we do this, we define some preliminary notation. Let $T_n = 2^{2^{n}}$, and let $\V_n$ be the set of types over $T_n$ rounds; in this section we will only consider mechanisms over $T_n$ rounds for some $n$. Let $\pi$ be a truthful mechanism over $\V_n$ (i.e. over $T_n$ rounds). Note that $\pi$ induces a truthful mechanism $\pi'$ over $\V_{n-1}$ in the following manner:

\begin{enumerate}
    \item Let $\vb_{n-1} = (v_1, v_2, \dots, v_{T_{n-1}})$ be a type in $\V_{n-1}$. From this, construct the following type $\vb_n \in \V_n$ via
    
    $$\vb_n = \left(\frac{v_1}{T_{n-1}}, \underset{T_{n-1} \mathrm{ times}}{\dots}, \frac{v_1}{T_{n-1}}, \frac{v_2}{T_{n-1}}, \underset{T_{n-1} \mathrm{ times}}{\dots}, \frac{v_2}{T_{n-1}}, \dots, \frac{v_{T_{n-1}}}{T_{n-1}}, \underset{T_{n-1} \mathrm{ times}}{\dots}, \frac{v_{T_{n-1}}}{T_{n-1}}\right)$$
    
    where each $v_t$ gets repeated $T_{n-1}$ times and multiplied by $1/T_{n-1}$. Note that since $T_{n-1}^2 = T_{n}$, this type indeed belongs to $\V_n$. Moreover, note that the total value $V_n$ of $\vb_n$ equals the total value $V_{n-1}$ of $\vb_{n-1}$. 
    
    \item
    Let $r_{\vb_{n}, t}$ and $x_{\vb_{n}, t}$ be the allocation and pricing rule for $\vb_{n}$ under $\pi$. 
    \item 
    Define the allocation and pricing rule for $\vb_{n-1}$ under $\pi'$ via
    
    $$r_{\vb_{n-1}, t} = \frac{1}{T_{n-1}}\sum_{s=1}^{T_{n-1}}r_{\vb_{n}, (t-1)T_{n-1}+s}$$
    
    and
    
    $$x_{\vb_{n-1}, t} = \sum_{s=1}^{T_{n-1}}x_{\vb_{n}, (t-1)T_{n-1}+s}.$$
    
\end{enumerate}

It is straightforward to verify that if $\pi$ is truthful then $\pi'$ is truthful. Note further that $\Rev(\vb_{n}) = \Rev(\vb_{n-1})$, so if $\pi$ satisfies

$$\Rev(\vb) \geq cV - d$$

\noindent
for all types $\vb \in \V_n$, then $\pi'$ also satisfies 

$$\Rev(\vb') \geq cV' - d$$

\noindent
for all types $\vb' \in V_{n-1}$. 

We will now prove the following analogue of Lemma \ref{lem:recurse}. Recall that $U_{\vb} = \sum_{t=1}^{T} (r_{\vb, t}v_t - x_{\vb, t})$ equals the total utility type $\vb$ receives in this mechanism.

\begin{lemma}\label{lem:recurse_discrete}
Let $\pi$ be a truthful mechanism over $\V_n$ that satisfies $\Rev(\vb) \geq cV - d$ for all $\vb \in \V_n$, and let $\pi'$ be the induced mechanism for $\V_{n-1}$. If, for some $\alpha >0$, $\beta \geq 0$,  

$$U_{\vb} \leq \alpha V + \beta$$ 

\noindent
for all types $\vb \in \V_n$, then

$$U_{\vb'} \leq \alpha' V + \beta'$$

\noindent
for all types $\vb' \in \V_{n-1}$, where

$$\alpha' = \alpha + \alpha \log \frac{1}{\alpha} - c$$

\noindent
and

$$\beta' = d + e^{1/\alpha}(\beta + 3) + 1.$$
\end{lemma}
\begin{proof}
We follow the rough structure of the proof of Lemma \ref{lem:recurse}. 

Let $\vb' \in \V_{n-1}$ be a type for $\pi'$, and let $\vb \in \V_{n}$ be the corresponding type for $\pi$. Let $\spike_{\vb}(t) \in \V_{n}$ (for $t \in [T_{n-1}]$) be the type defined via:

\begin{equation}
    \spike_{\vb}(t)_{s} = \begin{cases}
    v_s & \hbox{ if } s \leq tT_{n-1} \\
    \frac{1}{T_{n-1}}\sum_{r=t}^{T_{n-1}}v'_{r} & \hbox{ if }s \in [tT_{n-1}+1, (t+1)T_{n-1}] \\
    0 & \hbox{ otherwise}
    \end{cases}
\end{equation}

Consider the deviation where $\spike_{\vb}(t)$ pretends to be type $\vb$ for the first $(t+1)T_{n-1}$ rounds and then exists the protocol. The utility $U'_{\spike_{\vb}(t)}$ this type $\spike_{\vb}(t)$ receives from performing this deviation is at least the utility it receives during the deviation, i.e.

\begin{eqnarray*}
U'_{\spike_{\vb}(t)} &\geq& \sum_{s=tT_{n-1}+1}^{(t+1)T_{n-1}} \left(r_{\vb, s}\spike_{\vb}(t)_{s} - x_{\vb, s}\right) \\
&=& \sum_{s=tT_{n-1}+1}^{(t+1)T_{n-1}} \left(r_{\vb, s}\left(\frac{1}{T_{n-1}}\sum_{r=t}^{T_{n-1}}v'_{r}\right) - x_{\vb, s}\right) \\
&=& \left(\sum_{r=t}^{T_{n-1}}v'_{r}\right)\left(\frac{1}{T_{n-1}}\sum_{s=tT_{n-1}+1}^{(t+1)T_{n-1}} r_{\vb, s}\right)  - \left(\sum_{s=tT_{n-1}+1}^{(t+1)T_{n-1}} x_{\vb, s}\right) \\
&=& \left(\sum_{r=t}^{T_{n-1}}v'_{r}\right)r_{\vb', t}  - x_{\vb', t}.
\end{eqnarray*}

On the other hand, since $\pi$ is truthful, $U'_{\spike_{\vb}(t)}$ is less than the utility $U_{\spike_{\vb}(t)}$ this type would receive by playing truthfully, which by assumption is at most $\alpha V + \beta$. Therefore, for all $t \in [T_{n-1}]$, we have that

\begin{equation}\label{eq:preeq}
\alpha V + \beta \geq \left(\sum_{s=t}^{T_{n-1}}v'_{s}\right)r_{\vb', t}  - x_{\vb', t}.
\end{equation}

Since $x_{\vb', t} \leq 1$, we can relax this to:

\begin{equation}
\alpha V + (\beta + 1) \geq \left(\sum_{s=t}^{T_{n-1}}v'_{s}\right)r_{\vb', t}.
\end{equation}

Rewrite this further as

\begin{equation}
\frac{V}{r_{\vb', t}} \geq - \frac{(\beta + 1)}{\alpha r_{\vb', t}} + \frac{1}{\alpha}\left(V - \sum_{s=1}^{t}v'_{s}\right).
\end{equation}

\noindent
and even further as

\begin{equation}\label{eq:main_ineq}
\frac{V}{\alpha} \leq \frac{\beta+1}{\alpha r_{\vb', t}} + \frac{V}{r_{\vb', t}} + \frac{1}{\alpha}\sum_{s=1}^{t}\vb_{s}'.
\end{equation}

Let $w_{t} = r_{\vb, t}v'_{t}$. We now use inequality \ref{eq:main_ineq} to prove the following analogue of Lemma \ref{lem:hazardub}:

\begin{lemma}\label{lem:hazardub_discrete}
If $\tau \in [T_{n-1}]$ satisfies $\sum_{t=1}^{\tau} w_{t} \geq (\alpha \log \frac{1}{\alpha})V$, then

$$\sum_{t=1}^{\tau} \vb'_{t} \geq (1-\alpha)V - e^{1/\alpha}(\beta+3).$$
\end{lemma} 
\begin{proof}

Define

\begin{equation}
\lambda_{t} = \alpha \exp\left(\frac{\sum_{s=1}^{t}w_{s}}{\alpha V}\right)\frac{w_{t}}{V}.
\end{equation}

Multiplying inequality \ref{eq:main_ineq} by $\lambda_{t}$ and summing from $t=1$ to $\tau$, we have that

\begin{equation}\label{eq:main_ineq_sum}
\sum_{t=1}^{\tau}\frac{V}{\alpha}\lambda_{t} \leq \sum_{t=1}^{\tau}\frac{(\beta+1)\lambda_t}{\alpha r_{\vb', t}} + \sum_{t=1}^{\tau}\frac{V\lambda_t}{r_{\vb', t}} + \sum_{t=1}^{\tau}\frac{\lambda_t}{\alpha}\sum_{s=1}^{t}\vb_{s}'.
\end{equation}

We will examine each of these terms in turn. Before we begin, let $W_{t} = \sum_{s=1}^{t} w_{s}$. Then, note that

\begin{equation}
\sum_{t=1}^{\tau}\lambda_{t} = \sum_{t=1}^{\tau}\alpha \exp\left(\frac{\sum_{s=1}^{t}w_{s}}{\alpha V}\right)\frac{w_{t}}{V},
\end{equation}

can be viewed as a Riemann sum for 

$$\int_{0}^{W_{\tau}/V}\alpha e^{x/\alpha} dx = \alpha^2\left(\exp\left(\frac{W_{\tau}}{\alpha V}\right) - 1\right)$$

Since $\frac{w_t}{V} \leq \frac{1}{V}$ and since $\alpha e^{x/\alpha} \leq \alpha e^{1/\alpha}$, we therefore have that (for any $\tau$) 

\begin{equation}\label{eq:lambda-bound}
    \left| \left(\sum_{t=1}^{\tau} \lambda_t\right) - \alpha^2\left(\exp\left(\frac{W_{\tau}}{\alpha V}\right) - 1\right) \right| \leq \frac{\alpha e^{1/\alpha}}{V}.
\end{equation}

We now start by looking at the LHS of inequality \ref{eq:main_ineq_sum}. From inequality \ref{eq:lambda-bound}, we have that 

\begin{eqnarray*}
\sum_{t=1}^{\tau}\frac{V}{\alpha}\lambda_{t} &\geq& -e^{1/\alpha} + V \alpha \left(\exp\left(\frac{W_{\tau}}{\alpha V}\right) - 1\right) \\
&\geq& -e^{1/\alpha} + (1-\alpha) V
\end{eqnarray*}

\noindent
where the second inequality uses the fact that $W_{\tau} \geq (\alpha \log \frac{1}{\alpha})V$. 
The first term of the RHS of inequality \ref{eq:main_ineq_sum} can be bounded via

\begin{eqnarray*}
\sum_{t=1}^{\tau}\frac{(\beta+1)\lambda_t}{\alpha r_{\vb', t}} &= & \sum_{t=1}^{\tau}\frac{(\beta+1)w_{t}}{Vr_{\vb', t}} \exp\left(\frac{W_t}{\alpha V}\right) \\
&\leq & e^{1/\alpha}(\beta + 1) \sum_{t=1}^{\tau}\frac{w_{t}}{Vr_{\vb', t}} \\
&=& e^{1/\alpha}(\beta + 1) \sum_{t=1}^{\tau}\frac{v'_{t}}{V} \\
&\leq & e^{1/\alpha}(\beta + 1).
\end{eqnarray*}

The third term of the RHS of inequality \ref{eq:main_ineq_sum} can be written as

\begin{eqnarray*}
\sum_{t=1}^{\tau}\frac{\lambda_t}{\alpha}\sum_{s=1}^{t}\vb_{s}' &= & \sum_{s=1}^{t}\vb_{s}'\sum_{t=s}^{\tau}\frac{\lambda_t}{\alpha} \\
&\geq& \sum_{s=1}^{\tau}\vb_{s}'\left(\alpha \exp\left(\frac{W_{\tau}}{\alpha V}\right) - \alpha \exp\left(\frac{W_{s}}{\alpha V}\right)  + \frac{e^{1/\alpha}}{V}\right) \\
&\geq& \sum_{s=1}^{\tau}\vb_{s}'\left(\alpha \exp\left(\frac{W_{\tau}}{\alpha V}\right)\right) - \sum_{s=1}^{\tau}\vb_{s}'\left(\alpha \exp\left(\frac{W_{s}}{\alpha V}\right)\right) + e^{1/\alpha} \\
&\geq & \sum_{s=1}^{\tau} \vb'_{s} - \sum_{s=1}^{\tau}\frac{V\lambda_{s}}{r_{\vb', s}} + e^{1/\alpha}.
\end{eqnarray*}

Combining this all, we have that

\begin{equation*}
    -e^{1/\alpha} + (1-\alpha) V \leq  e^{1/\alpha}(\beta + 1) + \sum_{t=1}^{\tau}\frac{V\lambda_t}{r_{\vb', t}} + \sum_{s=1}^{\tau} \vb'_{s} - \sum_{s=1}^{\tau}\frac{V\lambda_{s}}{r_{\vb', s}} + e^{1/\alpha}
\end{equation*}

and simplifying, we have that

\begin{equation}
    \sum_{s=1}^{\tau} \vb'_{s} \geq (1-\alpha) V - e^{1/\alpha}(\beta+3).
\end{equation}

\end{proof}

Let $W_{\vb'} = \sum_{t} w_{t}$ and $X_{\vb'} = \Rev(\vb') = \sum_{t} x_{\vb', t}$; note that $U_{\vb'} = W_{\vb'} - X_{\vb'}$. By assumption, we know that $X_{\vb'} \geq cV - d$. Now, if there exists no $\tau$ such that $W_{\tau} \geq (\alpha \log \frac{1}{\alpha})V$, then this means that $W_{\vb'} \leq (\alpha \log \frac{1}{\alpha})$ and therefore

$$U_{\vb'} \leq \left(\alpha \log \frac{1}{\alpha} - c\right)V + d \leq \alpha'V + \beta'.$$

On the other hand, if there does exist a $\tau$ such that $W_{\tau} \geq (\alpha \log \frac{1}{\alpha})V$, pick the smallest such $\tau$; since $w_{t} \leq 1$, this means that $W_{\tau} \leq (\alpha \log\frac{1}{\alpha})V + 1$. We therefore have that (from Lemma \ref{lem:hazardub_discrete})  

\begin{eqnarray*}
V &=& \sum_{t=1}^{T_{n-1}} \vb'_{t} \\
&=& \sum_{t=1}^{\tau} \vb'_{t} + \sum_{t=\tau+1}^{T_{n-1}} \vb'_{t} \\
&\geq & (1-\alpha)V - e^{1/\alpha}(\beta + 3) + \sum_{t=\tau+1}^{T_{n-1}} w_{t} \\
&= & (1-\alpha)V - e^{1/\alpha}(\beta + 3) + W_{\vb} - W_{\tau} \\
&\geq& (1-\alpha)V + W_{\vb} - \left(\alpha \log\frac{1}{\alpha}\right)V - e^{1/\alpha}(\beta + 3) - 1.
\end{eqnarray*}

Rearranging this we have that

\begin{equation}
    W_{\vb} \leq \left(\alpha + \alpha\log \frac{1}{\alpha}\right)V + e^{1/\alpha}(\beta + 3) + 1,
\end{equation}

from which it immediately follows that $U_{\vb} = W_{\vb} - X_{\vb} \leq \alpha'V + \beta'$. 

\end{proof}

We can now proceed to prove Theorem \ref{thm:disc_ub}.

\begin{proof}[Proof of Theorem \ref{thm:disc_ub}]
Consider the function $f(\alpha) = (\alpha + \alpha\log(1/\alpha) - c)$. Since $c > 1/e$, (by the same logic as in the proof of Theorem \ref{thm:singlb}) there exists some finite $m$ such that $f^{(m)}(1) < 0$ (where $f^(m)$ is the function $f$ iterated $m$ times).

Now, consider a truthful mechanism $\pi$ for the single buyer game (for $\V_n$, for a sufficiently large $n$ that we will choose later) which statisfies $\Rev(\vb) \geq cV - d$. Note that this mechanism must satisfy $U_{\vb} \leq V$ for all types $\vb \in V$. Iterating Lemma \ref{lem:recurse} $m$ times, this implies there exists a truthful mechanism over $\V_{n-m}$ which satisfies $U_{\vb} \leq \alpha_{m}V + \beta_{m}$, where $\alpha_{m} = f^{(m)}(1) < 0$, and $\alpha_{m}$ and $\beta_{m}$ are both independent of $n$. 

Now, note that for sufficiently large $r$ (in particular, any $r$ where $T_{r} > -\beta_{m}/\alpha_{m}$), $\alpha_{m}T_{r} + \beta_{m} < 0$. But this means that if $V \geq T_{r}$, then $U_{\vb} < 0$, which cannot be true of any truthful mechanism. By taking $n = r + m$ (and looking at the all ones type $\vb \in \V_{r}$), we therefore arrive at a contradiction, and therefore no mechanism for the single buyer game over $\V_{n}$ can satisfy $\Rev(\vb) \geq cV - d$. 
\end{proof}

\subsection{Proof of upper bound for two players}\label{sect:2p-disc-ub}

In this section we rigorously prove Theorem \ref{thm:2plb}, showing that there is no limited liability mechanism for two bidders which guarantees an $\alpha > 2/e$ approximation to selling the business. Largely, we will follow the proof sketch given in Section \ref{sect:2bid}, the main complication being that in order to apply the minimax theorem we will have to discretize our type spaces and our protocols. 

To do this, we will also need to relax our notion of truthfulness slightly. Let $\tV_{T}$ be the set of type profiles $\vb$ for the 2 bidder game over $T$ rounds such that each $v_{i, t}$ is an integer multiple of $\eps_{T} = T^{-100}$. A discrete protocol is a protocol which specifies an allocation rule $r_{\vb, i, t}$, a payment rule $x_{\vb, i, t}$, and a reimbursement $g_{\vb, i}$ for each discrete type $\vb \in \tV_{T}$. We say that a protocol is $\eps_T$-truthful if it satisfies all the inequalities a truthful protocol must satisfy with some additional $o(1)$ slack. Specifically, 

\begin{itemize}
    \item \textit{(Limited liability)} For any $\mathbf{v} \in \tV_{T}$, $i$, and $t$, we must have $x_{\mathbf{v}, i, t} \leq r_{\vb, i, t}v_{i, t} + 10\eps_{T}$.
    \item \textit{(Incentive compatibility)} Fix a type profile $\mathbf{v} \in \tV_T$ and a bidder $i$, and let $\mathbf{v}' \in \tV_{T}$ be a profile where bidder $i$'s type $v_{i, t}$ is replaced by a new type $v'_{i, t}$. Then, if it is the case that
    
    $$x_{\mathbf{v}', i, t} \leq r_{\vb', i, t}v_{i, t} + 10\eps_{T}$$
    
    for all $t \leq \tau$ (i.e. our bidder with limited liability can pretend to have type $\mathbf{v}'$ for $\tau$ rounds) we must have that, for any $\tau \leq T$, 
    
    $$\left(\sum_{t=1}^{T} r_{\mathbf{v}, i, t}v_{i, t} - x_{\mathbf{v}, i, t}\right) + g_{\mathbf{v}, i} \geq \left(\sum_{t=1}^{\tau} r_{\mathbf{v}', i, t}v_{i, t} - x_{\mathbf{v}', i, t}\right) + g_{\mathbf{v}', i}\mathbf{1}(\tau = T) - 10\eps_{T}T.$$
\end{itemize}

We further say that a protocol is $\eps_T$-discrete if all values $r_{\vb, i, t}$, $x_{\vb, i, t}$, and $g_{\vb}, i$ are multiples of $\eps_T$. Let $\tPi_{T}$ be the set of all $\eps_{T}$-discrete, $\eps_{T}$-truthful protocols for type profiles in $\tV_{T}$. 

We first claim that if there exists a $c$-competitive truthful protocol, there also exists an $c$-competitive protocol in $\tPi_{T}$. 

\begin{lemma}\label{lem:disc_to_cont}
Assume there exists a truthful protocol $\pi$ for the 2 bidder game over $T$ rounds that satisfies (for some fixed $d$) for all $\vb$

$$\Rev_{\pi}(\vb) \geq c\Rev_{STB}(\vb) - d.$$

Then there exists a discrete protocol $\tilde{\pi}$ in $\tPi_{T}$ that satisfies for all $\vb \in \tV_T$

$$\Rev_{\tilde{\pi}}(\vb) \geq c\Rev_{STB}(\vb) - d - T^{-99}.$$
\end{lemma}
\begin{proof}
Define $\rnd_{\eps}(x) = \eps\lfloor \frac{x}{\eps} \rfloor$ (i.e. $\rnd_{\eps}(x)$ rounds $x$ down to the nearest multiple of $\eps$). Note that $|x - \rnd_{\eps}(x)| \leq \eps$.

Consider the allocation and payment rules $r_{\vb, i, t}$ and $x_{\vb, i, t}$ of $\pi$. We will use these to construct allocation and payment rules $\tilde{r}_{\vb, i, t}$ and $\tilde{x}_{\vb, i, t}$ for $\tilde{\pi}$ by letting (for each $\vb \in \tV_T$):

\begin{eqnarray*}
\tilde{r}_{\vb, i, t} &=& \rnd_{\eps_{T}}(r_{\vb, i, t}) \\
\tilde{x}_{\vb, i, t} &=& \rnd_{\eps_{T}}(x_{\vb, i, t}).
\end{eqnarray*}

By construction, $\tilde{\pi}$ is $\eps_{T}$ discrete. Since $|\tilde{r}_{\vb, i, t} - r_{\vb, i}| \leq \eps_{T}$ and $|\tilde{x}_{\vb, i, t} - x_{\vb, i, t}| \leq \eps_{T}$ for each $\vb \in \tV_{T}$, it follows immediately that $\tilde{\pi}$ satisfies the inequalities in the definition of $\eps_{T}$-truthful and therefore $\tilde{\pi} \in \tPi_{T}$. Finally, note that

$$\Rev_{\tilde{\pi}}(\vb) = \sum_{t=1}^{T} \tilde{x}_{\vb, i, t} \geq \sum_{t=1}^{T} x_{\vb, i, t} - T\eps_{T} = \Rev_{\pi}(\vb) - T\eps_{T} \geq c\Rev_{STB}(\vb) - d - T^{-99}.$$
\end{proof}

Specifically, to show that there is no $\alpha$-competitive truthful protocol, it suffices to show that there is no $\alpha$-competitive discrete protocol in $\tPi_{T}$. To do this, we will use our upper bound for the single bidder problem. Note that for the single bidder problem, we can similarly define the set of discrete types $\tVo_{T}$ (whose values are all multiples of $\eps_{T}$) and the set of discrete protocols $\tPio_{T}$ (which are $\eps_{T}$-truthful in the same sense as above, and whose allocation / pricing rules are also always multiples of $\eps_{T}$). Furthermore, write $\tVo_{T, V}$ to denote the discrete types with total value $V$ and $\tPio_{T, V}$ to denote the discrete protocols over types with total value $V$. 

Consider the following two-player zero-sum game. The first player picks a protocol $\pi$ in $\tPio_{T, V}$ and the second player picks a type $\vb$ in $\tVo_{T, V}$. The first player receives a score of $\Rev_{\pi}(\vb)$.

We first show that (as a consequence of our upper bound in the single buyer game) that for any $c > 1/e$, if the first player commits to their move first (even if they can commit to a mixed strategy), there is a type $\vb$ for the second player that guarantees $\Rev_{\pi}(\vb) \leq cV$. 

\begin{theorem}
Fix any $c > 1/e$ and $d \geq 0$. For all sufficiently large $T$, there exists a $V$ such that for any distribution $P_{T, V}$ over protocols in $\tPio_{T, V}$, there exists a $\vb \in \tVo_{T, V}$ such that

$$\E_{\pi \sim P_{T, V}}[\Rev_{\pi}(\vb)] \leq cV - d.$$
\end{theorem}
\begin{proof}
This follows directly from the proof of Theorem \ref{thm:disc_ub} with the caveats that Theorem \ref{thm:disc_ub} (1) applies to non-discrete types and protocols, (2) applies to truthful (i.e. not $\eps$-truthful) protocols, and (3) applies to single protocols and not distributions over protocols. 

We claim that all these concerns can be dealt with by examining the proof of Theorem \ref{thm:disc_ub}. Indeed for (1), note that (e.g. in the proof of Lemma \ref{lem:recurse_discrete}) if $\vb'$ is a discrete type then the constructed types $\vb$ and $\spike_{\vb}(t)_{s}$ are also discrete types, so the proof works as is with discrete types and protocols. For (2), note that the only consequence of the protocol being $\eps$-truthful (rather than truthful) is that we lose an additional additive $\eps_{T}T = T^{-99} \leq 1$ in inequality \ref{eq:preeq}, which slightly worsens the bound on $\beta'$ in Lemma \ref{lem:recurse_discrete} but does not otherwise affect the proof. For (3), note that the proof works with distributions over protocols as written.

(As a side remark, it is almost the case that this distribution $P_{T, V}$ over protocols gives rise to a truthful protocol $\bar{\pi}$ defined via $\bar{r}_{\vb, t} = \E_{\pi \sim P}[r_{\vb, t}]$ and $\bar{x}_{\vb, t} = \E_{\pi \sim P}[x_{\vb, t}]$. The one catch is that this distribution $\bar{\pi}$ might not be truthful because of the IC constraint -- it is not the case that if 

$$\bar{r}_{\vb', t}v_{t} \leq \bar{x}_{\vb', t}$$

\noindent
for all $t \in [0, \tau]$ that necessarily

$$r_{\vb', t}v_{t} \leq x_{\vb', t}$$

\noindent
for all $\pi = (r_{\vb, t}, x_{\vb, t})$ in the support of $P_{T, V}$ (although the converse is true). However, if we relax this constraint to $v_{t} \geq v'_{t}$ for all $t \in [0, \tau]$ (which implies this constraint and is all that the proof of Theorem \ref{thm:disc_ub} uses), then the expectation over truthful protocols is indeed a truthful protocol.)
\end{proof}

As a consequence of von Neumann's minimax theorem (since both the action spaces $\tPi_{T, V}$ and $\tV_{T, V}$ are finite), we have the immediate corollary.

\begin{corollary}\label{cor:minimax}
Fix any $c > 1/e$ and $d \geq 0$. For all sufficiently large $T$, there exists a $V$ such that there exists a distribution $\mu_{T, V}$ over $\tVo_{T, V}$ such that every protocol $\pi$ in $\tPio_{T, V}$ satisfies

$$\E_{\vb \sim \mu_{T, V}}[\Rev_{\pi}(\vb)] \leq cV - d.$$
\end{corollary}

We now proceed according to the proof sketch in Section \ref{sect:2bid}. Fix an $\alpha > 2/e$ and a $\beta \geq 0$, and assume to the contrary that for every $T$, we have a discrete protocol $\pi \in \tPi_{T}$ such that

$$\Rev_{\pi}(\vb) > \alpha\Rev_{STB}(\vb) - \beta$$

\noindent
for each $\vb \in \tV_{T}$. Let $c = \alpha/2 > 1/e$ and $d = \beta/2$. From Corollary \ref{cor:minimax}, choose a $T$ and $V$ (for this $c$ and $d$) and generate a hard distribution $\mu_{T, V}$. Construct a distribution $\mathcal{D}$ of types in $\tV_{2T}$ as follows. The valuation profile $\vb_1$ of the first bidder will be sampled so that $(v_{1, 1},  \dots, v_{1, T})$ is sampled according to $\mu_{T, V}$ and $v_{1, t} = 0$ for $t > T$. Likewise, $\vb_2$ will be sampled so that $(v_{2,T+1}, \dots, v_{2, 2T})$ is (independently from $\vb_1$) sampled according to $\mu_{T, V}$ and so that $v_{2, t} = 0$ for $t \leq T$. 

We will now use the behavior of $\pi$ on this distribution to construct two distributions $P_1$, and $P_2$ over single-buyer protocols in $\tPio_{T, V}$. To construct $P_1$ and $P_2$, we follow the following steps:

\begin{enumerate}
    \item Sample a distribution $\vb' \in \tVo_{T, V}$ from $\mu_{T, V}$.
    \item Define the function $\pad_{L}: \tVo_{T, V} \rightarrow \tVo_{2T, V}$ be the function that converts a discrete type $\vb \in \tVo_{T, V}$ over $T$ rounds to a discrete type over $2T$ rounds by \textit{prepending} $T$ zeroes to $\vb$. Likewise, let $\pad_{R}: \tVo_{T, V} \rightarrow \tVo_{2T, V}$ be the function that converts a discrete type $\vb \in \tVo_{T, V}$ over $T$ rounds to a discrete type over $2T$ rounds by \textit{appending} $T$ zeroes to $\vb$.
    \item From $\vb'$, we will use $\pi$ to generate a discrete mechanism $\pi_1 \in \tPio_{T, V}$. Assume that the mechanism $\pi \in \tPi_{2T}$ has allocation rules $r_{\vb^{(2)}, i, t}$ and $x_{\vb^{(2)}, i, t}$ where $\vb^{(2)} \in \tV_{2T}$; since $\vb'6{(2)}$ is the type profile for two bidders, we will write $\vb^{(2)} = (\vb_1, \vb_2)$. Then the allocation rules $(r_1)_{\vb, t}$ and $(x_1)_{\vb, t}$ for $\pi_1$ are given by (for $t \in [T]$):
    
    \begin{eqnarray*}
    (r_1)_{\vb, t} &=& r_{(\pad_{R}(\vb), \pad_{L}(\vb')), i, t}\\
    (x_1)_{\vb, t} &=& x_{(\pad_{R}(\vb), \pad_{L}(\vb')), i, t}
    \end{eqnarray*}
    
    The distribution $P_1$ is the distribution over these protocols $\pi_1$ induced by the randomness from sampling $\vb'$. 
    \item
    Likewise, we construct $\pi_2$ by setting
    
    \begin{eqnarray*}
    (r_2)_{\vb, t} &=& r_{(\pad_{R}(\vb'), \pad_{L}(\vb)), i, t}\\
    (x_2)_{\vb, t} &=& x_{(\pad_{R}(\vb'), \pad_{L}(\vb)), i, t}
    \end{eqnarray*}
    
    and construct $P_2$ similarly. 
\end{enumerate}

Since $\pi$ is $\eps_T$-discrete and $\eps_T$ truthful, each $\pi_1 \in \supp(P_1)$ and $\pi_2 \in \supp(P_2)$ are $\eps_T$-discrete and $\eps_T$-truthful (and thus belong to $\tPio_{T, V}$). Furthermore, note that (by the construction of $\mathcal{D}$, $P_1$, and $P_2$), we have that

\begin{equation}\label{eq:decomp}
    \E_{\vb^{(2)} \sim \mathcal{D}}[\Rev_{\pi}(\vb^{(2)})] = \E_{\vb_1 \sim \mu}[\E_{\pi_1 \sim P_1}[\Rev_{\pi_1}(\vb_1)]] + \E_{\vb_2 \sim \mu}[\E_{\pi_2 \sim P_2}[\Rev_{\pi_2}(\vb_2)]].
\end{equation}

Now, for each $\vb^{(2)} \in \supp(D)$, note that the value of both bidders is $V$, so $\Rev_{STB}(\vb^{(2)}) = V$, and thus 

$$\E_{\vb^{(2)} \sim \mathcal{D}}[\Rev_{\pi}(\vb^{(2)})] > \alpha V - \beta.$$

It follows from inequality \ref{eq:decomp} that either $\E_{\vb_1 \sim \mu}[\E_{\pi_1 \sim P_1}[\Rev_{\pi_1}(\vb_1)]]$ or $\E_{\vb_2 \sim \mu}[\E_{\pi_2 \sim P_2}[\Rev_{\pi_2}(\vb_2)]]$ is at least $\frac{1}{2}(\alpha V - \beta) = cV - d$. Without loss of generality, assume 

$$\E_{\vb_1 \sim \mu}[\E_{\pi_1 \sim P_1}[\Rev_{\pi_1}(\vb_1)]] > cV - d.$$

In particular, this implies there exists a $\pi_1^{*} \in \supp(P_1)$ such that 

$$\E_{\vb_1 \sim \mu}[\Rev_{\pi^{*}_1}(\vb_1)] > cV - d.$$

But this contradicts Corollary \ref{cor:minimax}. Combining this with Lemma \ref{lem:disc_to_cont}, we have the following lower bound.

\begin{theorem}
Fix any $\alpha > 2/e$, $\beta \geq 0$. Then, for sufficiently large $T$ there does not exist a truthful protocol for the 2 bidder game over $T$ rounds that satisfies (for all type profiles $\vb$)

$$\Rev(\vb) \geq \alpha \Rev_{STB}(v) - \beta.$$
\end{theorem}

\subsection{Proof of upper bound for $k > 2$ players}\label{sect:kp-disc-ub}

In this appendix, we prove Corollary \ref{cor:kbidderub}, showing that our upper bound of $2/e$ proved for $2$ players (Theorem \ref{thm:2plb}) extends to all numbers of players.

\begin{proof}
Assume to the contrary that there exists a truthful mechanism for $k$ bidders which is $\alpha$-competitive against selling the business for some $\alpha > 2/e$. We will use this to construct a truthful mechanism for $2$ bidders which is $\alpha$-competitive against selling the business, contradicting Theorem \ref{thm:2plb}. 

To do this, let $\vb_1$ and $\vb_2$ be arbitrary types, and introduce $k-2$ dummy types $\vb_{i}$ (for $3 \leq i \leq k$) all equal to the zero type (i.e. $v_{i, t} = 0$ for all $t \in [T]$). Note that $\Rev_{STB}(\vb_1, \vb_2, \vb_3, \dots, \vb_k) = \Rev_{STB}(\vb_1, \vb_2) = \min(|\vb_1|, |\vb_2|)$. Therefore, by running our $\alpha$-competitive truthful mechanism for $k$ bidders on these two real bidders and $k-2$ dummy bidders, we obtain an $\alpha$-competitive truthful mechanism for $2$ bidders, contradicting Theorem \ref{thm:2plb}.
\end{proof}

\subsection{Proof of mechanism for $k > 2$ players}\label{sect:kp-disc-lb}

In this appendix, we prove Theorem \ref{thm:kbiddermech-fpa}, showing that if we are given a lower bound $V^*$ on the highest total value of any bidder in a $k$-bidder game, the mechanism described in Section \ref{sect:kbidders} obtains a constant fraction of $V^*$ (independent of $k$) as revenue. 

\begin{proof}
Without loss of generality, assume $V_1 \geq V_2 \geq \dots \geq V_k$ (so $V_1 \geq V^*$). We will show that for early enough rounds $t$, it is a dominated strategy for the bidder with maximum value (bidder $1$) to bid less than $v_{1, t}/2$. Specifically, we have the following lemma:

\begin{lemma}\label{lem:maxstrat}
Let $W_i(t) = \sum_{s=1}^{t} v_{i, s}$ be the total value of bidder $i$ until time $t$. If round $t'$ satisfies $W_1(t') \leq 0.8V_1 - 1$, then any pure strategy for bidder $1$ with $b_{1, t'} < v_{1, t'}/2$ is a dominated strategy.
\end{lemma}
\begin{proof}
Let $s_1$ be a pure strategy (i.e. a choice of $V$ and $b_{t}$ to report) for bidder $1$ where $b_{1, t'} < v_{1, t'}/2$. Let $s'_1$ be a new pure strategy formed by modifying $s_1$ by increasing $b_{1,t'}$ to $b'_{1, t'} = v_{1, t'}/2$. We claim that (regardless of the strategies $s_{-1}$ of the other bidders), the utility $U(s'_1, s_{-i})$ bidder $i$ receives by playing $s'_1$ is at least the utility $U(s_1, s_{-1})$ bidder $1$ receives by playing $s_1$. 

These utility functions can be broken into three parts: the utility from being allocated the auction half of the item, the utility from being allocated the allocation half of the item, and the (negative) utility from paying the seller during the auction half. We will examine how these three components change between the two strategy profiles $(s_1, s_{-i})$ and $(s'_1, s_{-i})$.

To begin, note that if $t'$ takes place after the cap has triggered, then increasing the bid in this round can only possibly be beneficial for bidder $1$, since all payments after the cap has triggered are refunded. Thus, from now on we will assume that $t'$ takes place before the cap has triggered.

We start by examining the utility due to the allocation of the auction half. Note that since only the bid on round $t'$ changes between $s_1$ and $s'_1$, the allocation of the auction half remains the same for all rounds $t \neq t'$. For round $t'$, since bidder 1 increases their bid from strategy $s_1$ to $s'_1$, it is possible they were not allocated the item on round $t$ under $s_1$ but now are under $s'_1$. Strategy $s'_1$ thus gains more utility from the allocation of the auction half than $s_1$.

We next examine the difference in payments during the auction half. Let $X_1(t)$ be the total payment of bidder $1$ by time $t$ to the seller before the cap has been reached under strategy profile $(s_1, s_{-1})$. Let $X_1'(t)$ be the analogous quantity for strategy profile $(s'_1, s_{-1})$. Let $\Delta = X'_1(T) - X_1(T)$; note that this is the total difference in payments between strategy $s'_1$ and $s_1$ (since payments after the cap are refunded). To help bound this in terms of what we gain in the allocation half, we will prove the following facts about $X_1(t)$ and $X_1'(t)$:

\begin{enumerate}
    \item For all $t < t'$, $X'_1(t) = X_1(t)$.
    \item For all $t \geq t'$, $X'_1(t) - X_1(t)$ is decreasing in $t$.
    \item For all $t$, $X'_1(t) \geq X_1(t)$.
\end{enumerate}

The first fact holds since for all rounds $t < t'$, strategies $s_1$ and $s'_1$ perform the same actions and make the same payment. 

To see why the second fact is true, fix a round $t > t'$. Note that since $b_{1, t}$ is unchanged between $s_1$ and $s'_1$, if $b_{1, t}$ is the largest bid under $s_1$ it is also the largest bid under $s'_1$. The only way bidder $1$ pays a different amount in round $t$ under $s'_1$ than under $s_1$ is if the auction has already ended under $s'_1$ (i.e. the quantity $\sum \rho(X_i)$ has reached its cap of 1, or reaches this cap this round) but not yet under $s_1$. In this case bidder $1$ pays more under $s_1$ than $s'_1$ (and thus $X'_1(t) - X_1(t)$ is decreasing). 

To prove the third fact, note that after round $t'$, strategies $s_1$ and $s'_1$ will perform the same bids and payments until (possibly) the cap is triggered earlier under $s'_1$ (due to $s'_1$ bidding more in this round). Therefore, look at the first round $\tilde{t} > t'$ where the winning bid under $s_1$ and $s'_1$ differs. Let $X(t) = \sum_i X_i(t)$ be the total payment of all bidders $i$ under $s_1$ (likewise, let $X'(t) = \sum_i X_i'(t)$ be the total payment of all bidders $i$ under $s'_1$). For all $t \geq t'$ until this round, we must have that $X'(t) - X(t) = X'(t') - X(t')$ (since $t'$ is the only round where the behavior of $s_1$ and $s'_1$ differ). As mentioned above, the only way for the behavior under $s_1$ and $s'_1$ to differ is if the auction ends under $s'_1$ in round $\tilde{t}$; i.e., $X'(\tilde{t}) = V^*/10$. But this means that $X(\tilde{t}) \geq V^*/10 - (X'(t') - X(t'))$. In particular, it is impossible for $X_1(t)$ to increase by more than $(X'(t') - X(t'))$ after round $\tilde{t}$ (since this would make $X(t) \geq V^*/10$), and therefore $X'_1(t) - X_1(t) \geq 0$ for all $t \geq \tilde{t}$. 

Finally, we will examine the utility due to the allocation of the allocation half. In particular, we will show that strategy $s'_1$ gains at least $\Delta$ more utility from the allocation half than strategy $s_1$, thus compensating for the increased payment in the auction half. To do so, we will use the facts about $X'_1(t)$ and $X_1(t)$ that we have just proved. Note that under $s_1$, the utility bidder $1$ receives from the allocation half when playing according to $s_1$ is equal to

$$\sum_{t=1}^{T} \frac{1}{2}v_{1, t}\rho(X_1(t)).$$

Likewise the utility bidder $1$ receives from the allocation half when playing according to $s'_1$ is equal to

$$\sum_{t=1}^{T} \frac{1}{2}v_{1, t}\rho(X_1'(t)).$$

Therefore the extra utility bidder $1$ receives in the allocation half when playing according to $s'_1$ is at least

\begin{eqnarray*}
\sum_{t=1}^{T} \frac{1}{2}v_{1, t}(\rho(X_1'(t)) - \rho(X_1(t)) &= & \sum_{t=t'}^{T} \frac{1}{2}v_{1, t}(\rho(X_1'(t)) - \rho(X_1(t)) \\
&=& \sum_{t=t'}^{T} \frac{1}{2}v_{1, t}\frac{10(X_1'(t) - X_1(t))}{V^*} \\
&\geq& \sum_{t=t'}^{T} \frac{1}{2}v_{1, t}\frac{10\Delta}{V^*} \\
&=& 5\Delta(W_1(T) - W_1(t'-1))/V^* \\
&\geq& 5\Delta(V_1 - 0.8V_1 + 1)/V^* \geq \Delta \frac{V_1}{V^*} \geq \Delta.
\end{eqnarray*}

Here we have used the fact that since $X'_1(t) - X_1(t)$ is decreasing for $t \geq t'$, $X'_1(t) - X_1(t) \geq \Delta = X'_1(T) - X_1(T)$. It follows that $U(s'_1, s_{-1}) - U(s_1, s_{-1}) \geq \Delta - \Delta = 0$. This shows that $s_1$ is never better than $s'_1$; to see that it is dominated, it suffices to consider the strategy profile $s_{-1}$ where all other bidders bid $0$ every round. 
\end{proof}

Now, consider the total revenue received by the auction. Let $X_i(t)$ be the total amount bidder $i$ has paid to the auctioneer up until time $t$, and let $X(t) = \sum_{i} X_i(t)$. If we ever have that $X(t) = V^*/10$ (so we stop the auction), then we already have a $1/10$-competitive mechanism. Assume therefore that $X(t) < V^*/10$ for all $t$.

Consider the rounds $t$ where $W_1(t) \leq 0.8V_1 - 1$. By Lemma \ref{lem:maxstrat}, we know that bidder $1$ will bid at least $v_{1, t}/2$ in each such round. This means that the winning bid for each of these rounds is at least $v_{1, t}/2$ and the total revenue is therefore at least $(0.8V_1 - 2)/2$ (since the lowest the cumulative welfare can be in the last of these rounds is $.8V_1-2$, by fact that values are at most $1$). It follows that the total revenue is at least $0.4V_1 - O(1) \geq 0.4V^* - O(1)$, and once again, we receive a constant factor approximation (up to an additive $O(1)$ term) to $V^*$. The $O(1)$ reimbursement we provide to each player increases this additive loss to an additive $O(k)$. 
\end{proof}

We now prove Corollary \ref{cor:small_reimbursements}, showing that it is possible to implement the above mechanism with small $O(1)$-sized reimbursements. 

\begin{proof}[Proof of Corollary \ref{cor:small_reimbursements}]
We will make the following changes to the mechanism:

\begin{itemize}
    \item We will run the mechanism of Theorem \ref{thm:kbiddermech-fpa} with value $0.5V^*$ instead of $V^*$. We will also initialize each of the payments $X_i$ to $1$ at the beginning of the protocol; this only decreases the total revenue of the mechanism by at most an additive $O(k)$.
    \item Each round $t$, with probability $1/2$, instead of performing the action of the mechanism of Theorem \ref{thm:kbiddermech-fpa}, we instead run the following sub-procedure (a \textit{check}):
    \begin{itemize}
        \item Allocate an $\eps < 1/kT$ fraction of the item to the buyer.
        \item Expect bidder $i$ to pay $2\eps b_{i, t}$ this round.
        \item If the bidder pays less than this amount, eject them from the mechanism (never allocate to them again).
    \end{itemize}
\end{itemize}

We claim that under these changes, any strategy $s_i$ for bidder $i$ which submits some bids $b_{i, t}$ larger than $v_{i, t}/2$ is weakly dominated by another strategy. In particular, let $\tau$ be the last round for which $s_i$ submits a bid $b_{i, \tau}$ satisfying $b_{i, \tau} > v_{i, \tau}/2$. We claim $s_i$ is dominated by the strategy $s'_i$ for bidder $i$ which submits a bid $b'_{i, \tau} = v_{i, \tau}/2$ on round $\tau$ (and bids identically to $s_i$ on other rounds).

To see why, note that during a check round, if bidder $i$ is limited-liable and submitted a bid $b_{i, t} > v_{i, t}/2$, they are forced to defect; they cannot pay more than $\eps v_{i, t}$ for an $\eps$ fraction of the item by the limited-liability constraint. If they are forced to defect, they lose their $g > 1$ reimbursement for the rest of the game, as well as whatever future utility they may gain from the rest of the protocol (in particular, they lose at least their $10/V^*$ share of the remaining allocation half they received at the beginning of the game).

On the other hand, we can bound the expected utility that strategy $s'_i$ possibly loses due to decreasing their bid down to $v_{i, \tau}/2$. The possible negative effects of doing this are: (1) bidder $i$ might not win the item in round $\tau$ and (2) bidder $i$ might lose out on a $10/V^*$ share of the remaining allocation half. Since $v_{i, t}$ are bounded in size by at most $1$, the potential utility loss for $s'_i$ due to (1) is compensated for by the potential loss of the $g > 1$ reimbursement for $s_i$. Similarly, the potential utility loss for $s'_i$ due to (2) is at most $10/V^*$ times the remaining value of the allocation half of the item for bidder $i$ (i.e. $\sum_{t=\tau+1}^{T} v_{i, t}$), but $s_i$ loses at least their initial $10/V^*$ share of the allocation half. Thus overall the expected utility under $s_i$ is strictly less than the expected utility of $s'_i$, regardless of the actions of the other bidders. We can therefore conclude that bidders will not defect in this mechanism.

Since the highest bidder now has expected value at least $0.5V^*$ on the non-check rounds, the proof of Theorem \ref{thm:kbiddermech} shows that this mechanism receives a constant fraction of $0.5V^*$ as revenue, as desired.
\end{proof}

% \section{Randomized and fractional allocations}

% Throughout this paper, we make the simplifying assumption that the item for sale is fractionally divisible, i.e., that the seller can allocate any fraction $\rho$ of the item each round. For a non-divisible item, a common alternative is to 

% \section{Alternate posted-price model}

% As mentioned in Section MODEL, nn

% More specifically, consider the following alternate model. The initial setup is the same as before: there are $T$ rounds, $k$ bidders, and each bidder has adversarially set values $v_{i, t}$ for the item at round $t$. Now, in round $t$, instead of the seller allocating a fraction $r_{i, t}$ of the item to bidder $i$ directly, the seller posts a price $x_{i, t}$. ...

\end{document}